\documentclass[11pt]{article} 

\usepackage[utf8]{inputenc} 
\usepackage[T1]{fontenc} 

\usepackage{amsfonts, amsmath, amsthm, amssymb} 
\allowdisplaybreaks[1] 

\usepackage{graphicx} 
\usepackage{booktabs} 
\usepackage{wrapfig} 
\usepackage[top=0.8in,bottom=0.8in,left=0.8in,right=0.8in]{geometry} 
\usepackage[tiny]{titlesec} 
\usepackage[font=footnotesize,labelfont=bf]{caption}
\usepackage{epsfig}
\usepackage{color,url, placeins, enumitem}
\usepackage{array,blindtext,makecell} 
\usepackage{hyperref}
\usepackage{subfiles} 
\usepackage{bbm}
\usepackage{wrapfig}
\usepackage{multirow} 

\usepackage{cite}
\usepackage{float}

\usepackage{bm} 

\usepackage{enumitem}
\setlist[enumerate,1]{label=(\alph*)}

\newcommand{\od}{\stackrel{\mbox {\tiny {def}}}{=}}

\def\RR{\mathbb{R}}

\def\RR{\mathbb{R}}

\def\ZZ{\mathbb{Z}}

\def\S{{\mathcal{S}}}

\def\RR{\mathbb{R}}

\def\det{\operatorname{det}}
\def\S{{\mathcal{S}}}

\def\max{\mathrm{max}}

\def\supp{\operatorname{supp}}

\def\od{\stackrel{\mathrm{def}}{=}}

\def\supp{\operatorname{supp}}

\def\sgn{\operatorname{sgn}}
\def\FP{\operatorname{FP}}
\def\min{\operatorname{min}}

\def\idx{\operatorname{idx}}

\newcommand{\rotvert}{\rotatebox[origin=c]{90}{$\vert$}}

\theoremstyle{plain}
\newtheorem{theorem}{Theorem}
\newtheorem{corollary}{Corollary}
\newtheorem{lemma}{Lemma}

\theoremstyle{definition}
\newtheorem{definition}{Definition}

\definecolor{cherry}{rgb}{0.9,.1,.2}

\allowdisplaybreaks

\begin{document}
	
	\begin{center}
	{\large \bf Fixed point compositionality via low-rank gluing rules in \\ inhibition-dominated threshold-linear networks}\\
	\vspace{5pt}
	Juliana Londoño Álvarez\\
	{\small Division of Applied Mathematics, Brown University \\
		Robert J. and Nancy D. Carney Institute for Brain Science, Brown University}\\
	\vspace{10pt}
\end{center}
	
	\begin{abstract}
		Brains routinely generate highly flexible and complex behaviors on a relatively stable structure and limited resources. A key mechanism underlying this ability is compositionality, which allows the brain to efficiently decompose complex tasks into simpler, reusable primitives. While network modularity has often been linked to compositionality in biological and artificial networks, a rigorous mathematical characterization of this relationship in nonlinear networks is still lacking. In this work, we formally investigate how structural modularity supports functional compositionality in inhibition-dominated threshold-linear networks (TLNs). We introduce a novel class of modular network assembly called low-rank gluings, where component subnetworks with arbitrary internal connectivity are connected via specific low-rank couplings. We prove that the global fixed points of these networks are constrained to be combinations of the local fixed points of their constituent modules. For a more structured subclass, called rank-1 gluings, we provide a complete characterization that determines which combinations of local fixed points yield global ones. We apply these results to graph-based networks, extending fixed point decomposition rules from combinatorial threshold-linear networks (CTLNs) to the more flexible family of generalized CTLNs (gCTLNs), thereby proving that these structural rules are more robust than initially established. Finally, we demonstrate that these gluing rules provide a mathematically tractable recipe for engineering compositional dynamics, enabling the construction of networks with a combinatorially large repertoire of predictable attractors that can be understood from simpler component motifs, ranging from compositions of fixed points to compositional limit cycles.
		
		\medskip
		\noindent \textbf{Relevance to life sciences.} Although experimental evidence demonstrates that the brain relies on compositional strategies to flexibly recombine simple primitives into complex behaviors, the anatomical substrate underlying this ability remains unclear. In parallel, low-rank intercomponent connectivity has gained interest in neuroscience as a model for low-dimensional communication subspaces between distinct brain regions. Here, we bridge these ideas by demonstrating that low-rank intercomponent connectivity enables compositionality in modular, inhibition-dominated networks. 	This allows us to generate predictions about neural activity patterns during compositional tasks; model complex biological behavior (illustrated with a marine mollusk’s sensory network); and identify candidate architectures that support the brain's ability to generate a vast combinatorial repertoire of actions from a limited set of building blocks.

		\medskip
		\noindent \textbf{Mathematical content.} We develop this work using threshold-linear networks (TLNs), a class of continuous-time rate models widely used in computational neuroscience. While mathematically tractable, TLNs can still support a variety of complex nonlinear dynamics. We focus on inhibition-dominated networks, for which a solid theoretical foundation already exists. We apply these general results to graph-based architectures, where the connectivity matrix is defined by a directed graph. This extends already established theory from combinatorial threshold-linear networks (CTLNs), where all neurons share two global inhibition parameters, to the newer family of generalized CTLNs (gCTLNs), which permit heterogeneous, neuron-specific inhibition parameters. 
	\end{abstract}

	
	
	\section{Introduction}\label{sec:introduction}
	
	Understanding how network structure shapes function is a central question in the study of network dynamics across various disciplines \cite{strogatz2001,boccaletti2006}. In neuroscience, this question is of interest as the brain, despite having a relatively stable connectome, is capable of generating remarkably flexible and complex behaviors \cite{sporns2013a,park2013, bassett2017a,sporns2016,papo2025}. Recent work on connectome-constrained networks has shown that network modularity might be a key driver of this flexibility \cite{palma-espinosa2025}. Indeed, modularity is often linked to the ability of biological and artificial systems to decompose complex tasks into simpler primitives that can be flexibly recombined \cite{bakermans2025,reverberi2012}, a property commonly referred to as compositionality. Recent research on artificial neural networks suggests that multitasking networks naturally become modular during training, and that they reuse simple primitive dynamical motifs \cite{yang2019, driscoll2024, andreas2016,lepori2023}; similarly, experimental evidence suggests that the brain represents tasks compositionally by combining simpler neural representations \cite{willett2020,tafazoli2026,kymn2024,reverberi2012, bakermans2025,amematsro2025,he2026,tan2026}.
	
	However, despite strong evidence for a relationship between modularity and compositionality, a rigorous mathematical characterization of this connection is still lacking. While the interplay between structural modularity and network dynamics has long been an active area of research, most work has focused on the inverse problem of detecting modules from network activity \cite{sporns2016, lambiotte2022, schaub2017}. Studies on how modular structure shapes dynamics have largely been restricted to networks with linear dynamics \cite{lambiotte2022, sporns2000}, or either focus on macro-level statistical measures \cite{beiran2021a,palma-espinosa2025}, or routing of signals \cite{clark2025}. Consequently, there is a theoretical gap in our understanding of how modularity shapes the dynamics of nonlinear recurrent networks to support functional compositionality.
	
	In this work, we formally investigate this issue in the context of inhibition-dominated threshold-linear networks (TLNs). These are rate models with a long history in computational neuroscience \cite{tsodyks1999, hahnloser2003a, noorman2024, itskov2011, xie2002, hahnloser2003, curto2013, curto2016}, mathematically tractable yet capable of generating a rich diversity of nonlinear dynamics \cite{morrison2024}. We specifically focus on modular networks where component subnetworks have arbitrary internal connectivity, but intercomponent interactions are restricted to a simple type of low-rank connectivity. We refer to this coupling as a \textit{low-rank gluing} of the component subnetworks. While low-rank networks have been widely studied \cite{mastrogiuseppe2018a, schuessler2020, beiran2021a,clark2025}, low-rank intercomponent connectivity has not yet been addressed as a mechanism for compositionality. Here, we prove that the fixed points of low-rank gluings are constrained to be compositions of the fixed points of the individual parts. Moreover, for a more structured subclass called \textit{rank-1 gluings}, we provide a complete characterization that determines which compositions of local fixed points yield global ones. 
	
	By applying our theory to graph-based networks, we observe that the emergent dynamic attractors of glued networks also behave as compositions of the attractors of the subnetworks. Thus, this family of networks efficiently exhibits a rich and predictable repertoire of attractor dynamics by reusing simple component motifs. Since the introduction of Hopfield networks \cite{hopfield1982}, computing with attractors has been central to neural modeling, yet most applications remain limited to fixed points \cite{hahnloser2003a, noorman2024,seung1996,samsonovich1997, battaglia1998, tsodyks1999, burak2009, itskov2011, gardner2022, mcnaughton2006a, cohen1983, xie2002}. Our results provide a principled recipe to build architectures with combinatorially many attractors that range from compositions of discrete fixed points to compositional limit cycles, expanding the modeler's toolbox, and suggesting a principle by which biological networks might generate a range of complex behaviors from a limited set of building blocks.
	
	\paragraph{Roadmap.} The paper is organized as follows. Section~\ref{sec:results} presents our main results, beginning with two fixed point compositionality theorems for low-rank gluings in Section \ref{sec:main-thms}. We then apply these results to graph-based architectures to establish two additional compositionality theorems in Section \ref{sec:CTLN-app}, one of which is applied to modeling the sensory network of a marine mollusk during hunting in Section \ref{sec:mollusk-app}. Proofs of the main theorems are deferred to Section \ref{sec:proofs}. Finally, Section \ref{sec:discussion} discusses the broader implications of these findings and opportunities for future work.
	
	\noindent\rule{\textwidth}{1pt}
		\textbf{Box 1. Context on the mathematical and computational methods}
		\medskip
		
		An attractor network is a dynamical system composed of recurrently connected units whose state evolves over time until it settles into a stable, long-term pattern of activity known as an attractor. These emergent patterns can be static (such as fixed points) or dynamic (such as limit cycles and quasiperiodic or chaotic trajectories) \cite{eliasmith2007}. In computational neuroscience, static attractors are often used to model persistent representations (e.g., ring attractors for head direction or bump attractors for spatial memory), whereas dynamic attractors are ideal for modeling repeating behaviors (e.g., locomotion and respiration) \cite{eliasmith2007,khona2022,rolls2010}.
		
		While networks of biophysically detailed model neurons (e.g., Hodgkin-Huxley, integrate-and-fire, Morris-Lecar, etc.) can exhibit attractor dynamics \cite{lundqvist2010,solanka2015,krishnamurthy2012,dempere-marco2012,conklin2005,kang2019,wang2023}, attractor networks have historically been associated with networks built from simple artificial neurons \cite{amit1989}. These units sum their weighted inputs and pass the result through an activation function. In continuous time, these systems take either of the general forms:
		\begin{equation}\label{eq:attr-net}
			\tau_i \frac{dx_i}{dt} = -x_i + f\left(\sum_{j=1}^n W_{ij}x_j\right) \quad \text{ or } \quad \tau_i \frac{dx_i}{dt} = -x_i + \sum_{j=1}^n W_{ij}f\left(x_j\right), \quad i = 1, \dots, n,
		\end{equation}
		where $W$ is a connectivity matrix and $f$ is an activation or threshold function.
		
		Much theoretical research has focused on how specific connectivity structures shape the network's attractor landscape. For instance, certain classes of symmetric networks ($W_{ij} = W_{ji}$) are mathematically guaranteed to converge to static attractors \cite{hopfield1982,hahnloser2003,cohen1983}. In contrast, non-symmetric networks can give rise to richer, non-steady state attractor dynamics \cite{morrison2024,sompolinsky1988,curto2023}. In this work, we specifically focus our analyses on asymmetric connectivity.
		
		The complexity of these networks (and their analytical tractability) also depends on the activation function $f$. Linear networks represent the simplest possible case, where the well-developed theory of linear systems applies. In this paper, we focus on the next step in complexity: the threshold-linear activation function $[\cdot]_+ = \max\{0,\cdot\}$, which makes the system in Equation \eqref{eq:attr-net} into a piecewise linear system of ODEs. This simplifies mathematical analyses while remaining fundamentally nonlinear and able support a variety of complex attractor behaviors.
		
		With the advent of high-density multi-neuron recordings, the view that the brain operates as a high-dimensional dynamical system has gained widespread support. In this context, low-dimensional ``neural manifolds'' frequently observed in data can be seen as signatures of attractor dynamics. Consequently, modeling with attractor networks has again become a popular paradigm for understanding brain function \cite{langdon2023,perich2025}.
		\medskip
		
		\textbf{Further reading.} 
		\begin{itemize}
			\item \textbf{Dynamical systems in neuroscience:} \cite{abbott2001,izhikevich2006,ermentrout2010} are classical textbooks extensively covering the foundations of these topics. For more modern, concise reviews see \cite{breakspear2017,miller2016,rabinovich2006}.
			\item \textbf{Attractor networks and neural manifolds:} \cite{amit1989} is the classical textbook on artificial neuron networks. For a modern review on the use of attractor networks in brain modeling see \cite{khona2022}. For reviews on neural manifolds see \cite{langdon2023, perich2025}.
			\item \textbf{Theoretical foundations of this work:} \cite{curto2019a} explores the dynamics of threshold-linear networks (TLNs), while \cite{morrison2024,curto2023} focus more on combinatorial threshold-linear networks (CTLNs).
		\end{itemize}
	\noindent\rule{\textwidth}{1pt}
	
	\section{Results}\label{sec:results}
	
	We investigate the relationship between modularity and compositionality in the context of threshold-linear networks (TLNs). TLNs are a class of neural network models describing the firing rates $x_i(t)$ of $n$ recurrently connected neurons according to a system of ordinary differential equations (ODEs):
	\begin{equation}\label{eq:TLN-dynamics-gral}
		\dfrac{dx_i}{dt} = -x_i + \left[\sum_{j=1}^n W_{ij}x_j+b_i \right]_+, \quad i = 1,\ldots,n.
	\end{equation}
	
	The activation function, $[\cdot]_+ = \max\{0,\cdot\}$ is piecewise linear, which greatly simplifies the mathematics while allowing the system to still support nonlinear dynamics. The matrix $W =[W_{ij}] \in \mathbb{R}^{n \times n}$ prescribes the interactions among neurons. We focus on \textit{inhibition-dominated} (or competitive) networks, where $W_{ij}\leq 0$ for all $i,j = 1, \dots, n$. Competitiveness is known to facilitate oscillations in many systems \cite{whittington2000}, particularly in biological networks \cite{karnani2014,arriaga2017,fino2011,bucher2009}. Finally, $b_i$ represents the external input. We assume a uniform, constant external input $b_i = \theta > 0$. This assumption isolates the role of the connectivity matrix $W$ in generating the emergent dynamics. Under these constraints, a given network is fully specified by the choice of $W$ and $\theta$, and thus we denote a given TLN by $(W,\theta)$. 
		
		For nonlinear dynamical systems, important insights about their dynamics can be gained from investigating their fixed points. For TLNs, thanks to its piecewise linear nature, fixed points can be tracked via their supports (the set of ``ON'' neurons) \cite{curto2019a}:
		\begin{equation}\label{eq:FP-TLN}
			\FP(W,\theta) \od \{\supp{x^*} \subseteq [n] ~|~  x^* \text{ is a fixed point of } (W,\theta) \},
		\end{equation} where  $\supp{x^*} \od \{i \in [n] \mid x^*_i>0\}$ and $[n] \od \{1,\dots,n\}$\footnote{Under mild non-degeneracy conditions on $W$ and $b$, a TLN $(W,b)$ can have at most one fixed point per support, and the coordinates of the fixed point can be easily recovered from its support $\sigma \in \FP(W,\theta)$ as $\boldsymbol{x}_\sigma^*=\left(I-W_\sigma\right)^{-1} \boldsymbol{\theta}_\sigma$, where $x_k^* = 0$ for all $k \notin \sigma$ \cite{curto2019a}.}.
	
		The set $\FP(W,\theta)$ contains supports for both stable and unstable fixed points, and provides valuable information about the attractor landscape of the network. Stable fixed points are, by definition, static attractors; and, as we will see in Section \ref{sec:fp-correspondance}, certain unstable fixed points are closely related to dynamic attractors.
	
	\subsection{Fixed point compositionality}\label{sec:fp-compositionality}

	To understand fixed point compositionality in modular TLNs, let us distinguish between local and global dynamics. Consider a modular TLN $(W, \theta)$ consisting of component modules $(W_1, \theta), \dots, (W_N, \theta)$. We call \textit{local} fixed points (and dynamics) those of an isolated subnetwork $(W_i, \theta)$ in the absence of interactions with other components. On the other hand, we call \textit{global} fixed points (and dynamics) those of the whole coupled network $(W, \theta)$.  
	
	In this context, fixed point compositionality refers to the property where global fixed points are obtained as combinations of local fixed points. For example, the network in Figure \ref{fig:fixed_point_compositionality_and_attractors}D is composed of the three modules shown in Figure \ref{fig:fixed_point_compositionality_and_attractors}A-C, whose local fixed points are given by:
	\begin{align*}
		\FP(W_1, \theta) &= \{ {\{4\}}, {\{1, 2, 3\}}, \{1, 2, 3, 4\} \} \\
		\FP(W_2, \theta) &= \{ {\{8\}}, {\{5, 6, 7\}}, \{5, 6, 8\} \} \\
		\FP(W_3, \theta) &= \{ {\{11\}}, {\{9, 10\}}, \{9, 10, 11\} \}.
	\end{align*} \vspace{-10pt}
		\begin{figure}[!h]
		\begin{center}
			\includegraphics[width=\textwidth]{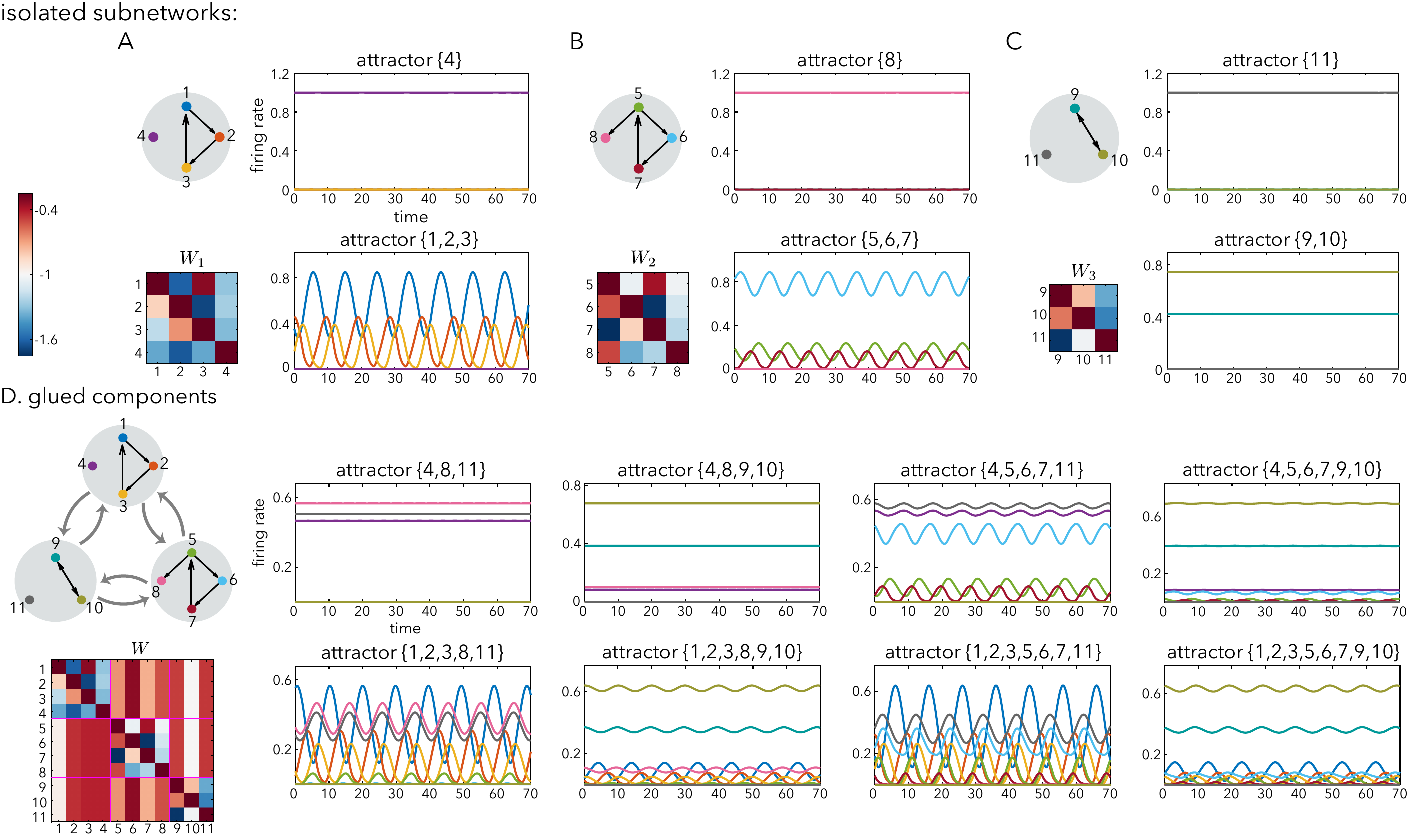}
		\end{center}
		\caption{\textbf{Fixed point compositionality and attractors.} Three component networks are connected via a low-rank gluing; local attractors form the building blocks of global attractors. The graphs display only the weak, less inhibitory edges $W_{ij} > -1$ to highlight the underlying functional structure. 
			(A-C) Isolated component networks, and their (local) attractors labeled by the highest-firing neurons.
			(D) A low-rank gluing of the components in panels A-C, and its global attractors labeled by the highest-firing neurons. Thick gray arrows between components represent all-to-all weak inhibitory connections between nodes in different components, with a single synaptic weight for all edges out of a given node, but different weights allowed across nodes. Each network on panels A-C has two attractors, combining to give the $2*2*2 = 8$ attractors of the network of panel D.
		}
		\label{fig:fixed_point_compositionality_and_attractors}
	\end{figure}	
	
 When networks are coupled via the low-rank intercomponent connectivity shown in the matrix of Figure \ref{fig:fixed_point_compositionality_and_attractors}D, where each node in a given component projects identically to every node outside its component, we find that all global supports $\sigma \in \FP(W, \theta)$ are formed as the union of exactly one local support per component. For instance, the global fixed point support $\{4, 8, 11\}$ is the union of the local fixed point supports $\{4\}$, $\{8\}$, and $\{11\}$. This yields a total of $3*3*3=27$ global fixed point supports:
 	\begin{multline*}
 	\FP(W, \theta) = \{ {\{4, 8, 11\}}, {\{4, 8, 9, 10\}}, {\{1, 2, 3, 8, 11\}}, {\{4, 5, 6, 7, 11\}}, {\{4, 5, 6, 8, 11\}}, \{4, 8, 9, 10, 11\}\\
 	\{1, 2, 3, 4, 8, 11\}, {\{1, 2, 3, 8, 9, 10\}}, {\{4, 5, 6, 7, 9, 10\}}, {\{4, 5, 6, 8, 9, 10\}}, \{1, 2, 3, 4, 8, 9, 10\}, \\
 	{\{1, 2, 3, 5, 6, 7, 11\}}, {\{1, 2, 3, 5, 6, 8, 11\}}, \{1, 2, 3, 8, 9, 10, 11\}, \{4, 5, 6, 7, 9, 10, 11\}, \{4, 5, 6, 8, 9, 10, 11\}, \\
 	\{1, 2, 3, 4, 5, 6, 7, 11\}, \{1, 2, 3, 4, 5, 6, 8, 11\}, \{1, 2, 3, 4, 8, 9, 10, 11\}, {\{1, 2, 3, 5, 6, 7, 9, 10\}}, \\
 	{\{1, 2, 3, 5, 6, 8, 9, 10\}}, \{1, 2, 3, 4, 5, 6, 7, 9, 10\}, \{1, 2, 3, 4, 5, 6, 8, 9, 10\}, \{1, 2, 3, 5, 6, 7, 9, 10, 11\}, \\
 	\{1, 2, 3, 5, 6, 8, 9, 10, 11\}, \{1, 2, 3, 4, 5, 6, 7, 9, 10, 11\}, \{1, 2, 3, 4, 5, 6, 8, 9, 10, 11\} \}.	
 \end{multline*}
 
 In this specific example, all fixed points obtained in this fashion yield global fixed points. In general, however, our work here shows that while all global fixed point supports of low-rank gluings must be unions of local fixed point supports, it is not always the case that all unions of local fixed point supports will yield global fixed point supports. 
 
 Remarkably, even though the theorems we present here concern exclusively the fixed points of the network, these decomposition properties often empirically carry over to the static and dynamic attractors of the glued network \cite{parmelee2022}. Throughout this manuscript, we use terms like ``attractor'' or ``limit cycle'' informally to refer to computationally observed stable trajectories to which many initial conditions converge. Mathematically proving that these trajectories are attractors is challenging \cite{bel2021}, and lies beyond the scope of this work. 
 
 This robust computational correspondence between fixed point structure and attractors allows us to observe compositionality in the global attractors, as demonstrated in Figure \ref{fig:fixed_point_compositionality_and_attractors}. By systematically varying initial conditions for the component networks of Figure \ref{fig:fixed_point_compositionality_and_attractors}A–C, we find that each subnetwork has two attractors, which we label by the fixed point support matching their highest-firing neurons. Correspondingly, for the glued network in Figure \ref{fig:fixed_point_compositionality_and_attractors}D, we similarly find eight attractors, all which appear to be assembled from local pieces, even when these local pieces are limit cycles and not simply stable fixed points. For example, global attractors whose labels contain the support $\{1, 2, 3\}$, corresponding to a local limit cycle (Fig. \ref{fig:fixed_point_compositionality_and_attractors}A), also exhibit limit cycle dynamics in the glued network (Fig. \ref{fig:fixed_point_compositionality_and_attractors}D, bottom row). 
 
 Note, however, that while the highest-firing neurons match across global attractors, neurons corresponding to local static attractors do not always remain static in the glued network, as exemplified by the global attractor $\{1,2,3,5,6,7,9,10\}$ (Fig. \ref{fig:fixed_point_compositionality_and_attractors}D, bottom right). Additionally, while the supports compose cleanly as unions, the numerical values of the local fixed points $x^*$ are not simply concatenated, as is clearly illustrated for the global attractor $\{4, 8, 11\}$ (Fig. \ref{fig:fixed_point_compositionality_and_attractors}D, top left).
	
	\subsection{Main theorems: low-rank gluings}\label{sec:main-thms}
	
	This section presents two of the four main theorems proven in this work. These first two theorems address fixed point compositionality in low-rank gluings for inhibition-dominated TLNs. The second two, presented in Section \ref{sec:CTLN-app}, focus on graph-based TLNs and require, in addition to the application of the theorems in this section, the development of further theory specific to graph-based networks. 
	
	We use the subscript notation for matrices and vectors to denote restriction to a specific subset of indices, so for instance $\mathbbm{1}_{\tau}$ denotes a column vector of ones indexed by the elements of $\tau$.  With this notation, we formally define low-rank gluings as follows:
	\begin{definition}\label{def:low-rank-gluings}
		We say that a TLN $(W,\theta)$ on $n$ nodes is a \emph{low-rank gluing} of $(W_1,\theta), \dots, (W_N,\theta)$ if there is a partition $\tau_1, \dots, \tau_N$ of $[n]$ such that $W_{\tau_i,\tau_i} = W_i$ for each $i\in[N]$, and $W_{\tau_i,\tau_j} = \mathbbm{1}_{\tau_i} \bm{\gamma}^\intercal$ for $\bm{\gamma} \in \mathbb{R}^{|\tau_j|}$. If, furthermore, $W_{([n]\setminus \tau_i),\tau_i} = \mathbbm{1}_{[n]\setminus \tau_i} \bm{\gamma}^\intercal$  for $\bm{\gamma} \in \mathbb{R}^{|\tau_i|}$, for each $i\in[N]$, then we say that $(W,\theta)$ is a \emph{rank-1 gluing}. 
	\end{definition} 
	
	The constraint $W_{\tau_i,\tau_j} = \mathbbm{1}_{\tau_i} \bm{\gamma}^\intercal$ for $\bm{\gamma} \in \mathbb{R}^{|\tau_j|}$ implies that the pairwise off-diagonal blocks of $W$ have constant columns (Fig. \ref{fig:low-rank-gluings-comparison}A). This means that any given node projects identically to every node within a target component, though the strength of these projections may vary across different components. On the other hand, the additional constraint $W_{([n]\setminus \tau_i),\tau_i} = \mathbbm{1}_{[n]\setminus \tau_i} \bm{\gamma}^\intercal$  for $\bm{\gamma} \in \mathbb{R}^{|\tau_i|}$ of rank-1 gluings means that a node must project identically to all nodes outside its own component. Consequently, each column of the combined off-diagonal blocks is constant (Fig. \ref{fig:low-rank-gluings-comparison}B). Note that Definition \ref{def:low-rank-gluings} allows individual subnetworks to have arbitrary internal connectivity $W_i$. 
		\begin{figure}[!h]
		\begin{center}
			\includegraphics[width=\textwidth]{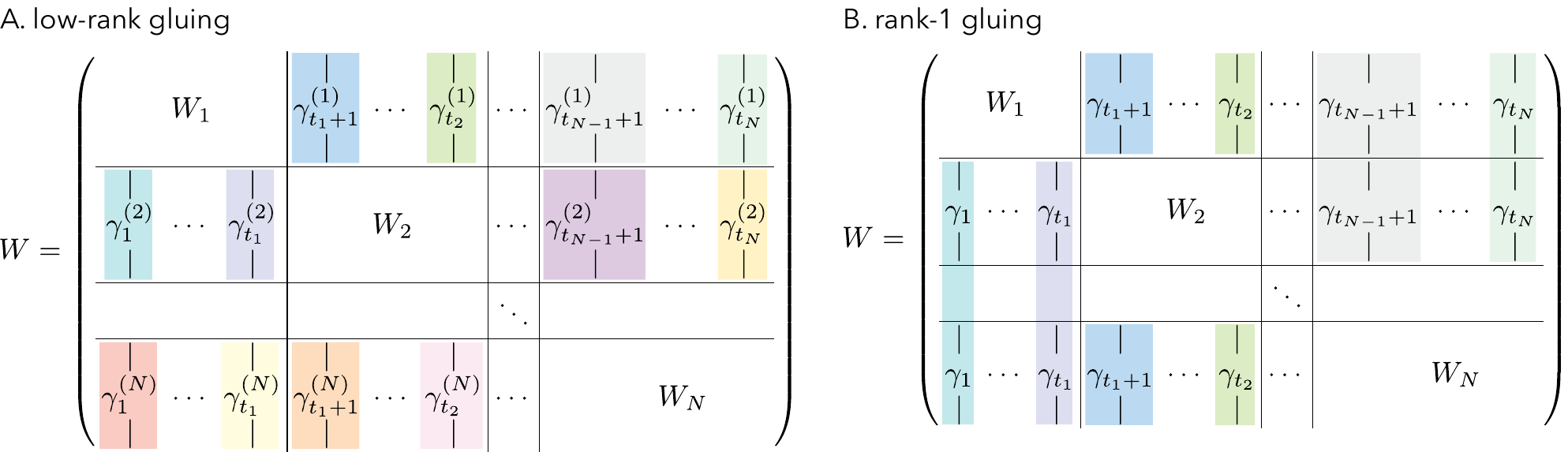}
		\end{center}
		\caption{\textbf{Low-rank gluings versus rank-1 gluings.} Elements with the same shading are constrained to be equal.
			(A) A connectivity matrix defining a low-rank gluing of component subnetworks $W_1, \dots, W_N$. Here, the scalar $\gamma_ j^{(i)}$ gives the strength of the synaptic connection from node $j$ to all nodes in the component subnetwork $W_i$.
			(B) A connectivity matrix defining a rank-1 gluing of component subnetworks $W_1, \dots, W_N$.}
		\label{fig:low-rank-gluings-comparison}
	\end{figure}
	
	While requiring the off-diagonal blocks of $W$ to take the form $\mathbbm{1}_{\tau} \bm{\gamma}^\intercal$ makes them into rank-1 submatrices, this structure is significantly more restrictive than merely having the off-diagonal blocks have rank 1 (of the general form $\bm{\beta} \bm{\gamma}^\intercal$). 
	
	Our first main result shows that low-rank gluings restrict the global fixed point supports to be unions of local fixed point supports:
	
	\begin{theorem}[fixed point compositionality]\label{thm:low-rank-gluing-menu} 	
		Let $(W,\theta)$ be a low-rank gluing of TLNs $(W_1,\theta), \dots, (W_N,\theta)$. Then, for each $\sigma \in \FP(W,\theta)$, $\sigma$ has the form: 
		$$\sigma = \bigcup_{\ell=1}^{N} \sigma_\ell, \text{ where } \sigma_\ell \in \FP(W_\ell,\theta) \cup \{\emptyset\}.$$ That is, the fixed point supports of $(W,\theta)$ are unions of the fixed point supports of component networks, at most one per component.
	\end{theorem}
	
	Observe that this theorem constrains the potential number of global fixed points of a TLN on $n$ nodes $(W,\theta)$ from the $2^n-1$ possible non-empty subsets of $n$ neurons to a maximum of $\prod_{i=1}^N (|\FP(W_i,\theta)|+1)-1$. For instance, the network in Figure \ref{fig:fixed_point_compositionality_and_attractors}D is a rank-1 gluing that has $27$ out of the possible $(3+1)*(3+1)*(3+1)-1  = 63$ fixed points allowed for a low-rank gluing defined in the components of  Figure \ref{fig:fixed_point_compositionality_and_attractors}A-C.  This is orders of magnitude smaller than the $2^{11}-1 = 2047$ combinations otherwise possible for a network on 11 nodes.
	
	To illustrate Theorem \ref{thm:low-rank-gluing-menu} more generally, consider the networks in Figure \ref{fig:menu-thm-low-rank-gluings-example}. Any low-rank gluing of the networks in panel A is guaranteed to produce a network whose set of fixed point supports only contains supports from the set $\mathcal{M}$ in panel B. As the examples in panels C and D show, this might range from as few as a single global fixed point (panel C, right) to the full set $\mathcal{M}$ (panel D, right). Which supports constitute global fixed points is determined by the specific intercomponent connections.	
	\begin{figure}[!h]
		\begin{center}
			\includegraphics[width=0.95\textwidth]{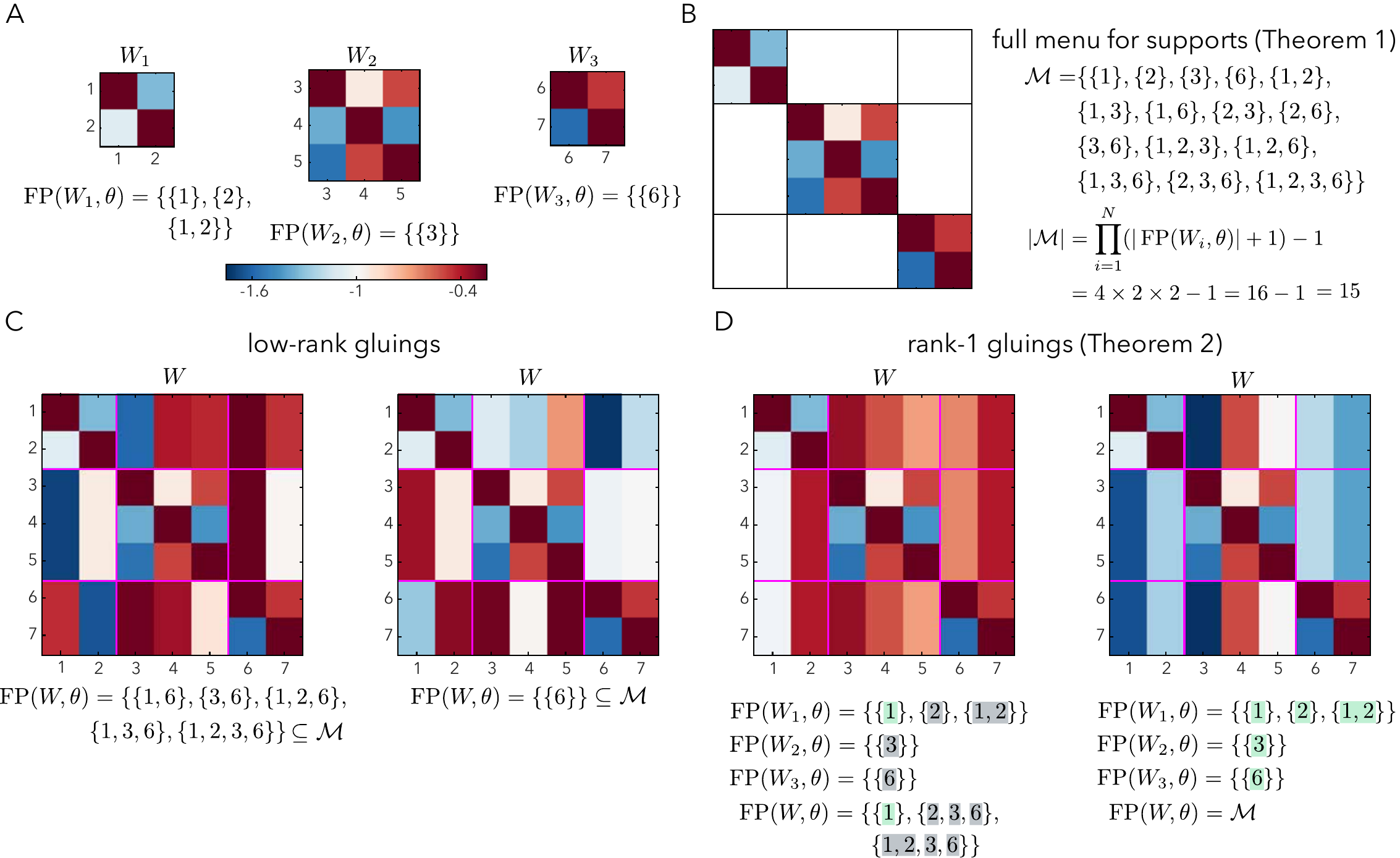}
		\end{center}
		\caption{\textbf{Illustration of Theorems \ref{thm:low-rank-gluing-menu} and \ref{thm:rank-1-gluings}.} 
			(A) Three component networks, and their fixed point supports $\FP(W_i,\theta)$.
			(B) Structure of a modular network assembled with the components of panel A, and the full set of possible supports when networks are low-rank gluings of those. 
			(C-D) Two example low-rank and rank-1 gluings of the components in panel A, and their fixed point supports $\FP(W,\theta)$. In panel C, Theorem \ref{thm:low-rank-gluing-menu} constrains $\FP(W,\theta)$ but it does not fully determine it. In panel D, Theorem \ref{thm:rank-1-gluings} fully determines $\FP(W,\theta)$. Local fixed points that survive in each rank-1 gluing are shaded green, while the ones that die are shaded gray.}
		\label{fig:menu-thm-low-rank-gluings-example}
	\end{figure}
	
	Theorem \ref{thm:low-rank-gluing-menu} provides a ``menu'' of possible supports, constraining all fixed points of glued networks $(W,\theta)$ to be unions of local fixed points. However, it does not guarantee which unions will actually be realized as global fixed points. To achieve a full characterization, we require the additional structure of a rank-1 gluing:
	\begin{theorem}\label{thm:rank-1-gluings}
		Let $(W,\theta)$ be a rank-1 gluing of $(W_1,\theta), \dots, (W_N,\theta)$.  Then $\sigma \in \FP(W,\theta)$ if and only if $\sigma = \bigcup_{\ell=1}^{N} \sigma_\ell$ for $\sigma_\ell \in \FP(W_\ell,\theta) \cup \{\emptyset\}$, and either
		\begin{enumerate}
			\item none of the $\sigma_\ell$ are in $\FP(W,\theta) \cup \{ \emptyset\}$, or
			\item every $\sigma_\ell$ is in $\FP(W,\theta) \cup \{ \emptyset\}$.
		\end{enumerate}
		In other words, $\sigma \in \FP(W,\theta)$ if and only if $\sigma$ is either a union of dying fixed points, exactly one from every component, or it is a union of surviving fixed points $\sigma_i$, at most one per component.
	\end{theorem}
	
	This theorem explicitly identifies which unions of local supports constitute global fixed points based on whether local fixed points ``die'' -- $\sigma_\ell \notin \FP(W,\theta)$, or ``survive'' -- $\sigma_\ell \in \FP(W,\theta)$ in the glued network. For example, in the left matrix of Figure \ref{fig:menu-thm-low-rank-gluings-example}D, only fixed point $\{1\}$ survives (green), while others die (gray). Thus, the global fixed points are given by unions of surviving fixed points, at most one per component ($\{1\}$) or unions of dying fixed points, exactly one from every component ($\{2,3,6\}, \{1,2,3,6\}$). In the second matrix of Figure \ref{fig:menu-thm-low-rank-gluings-example}D, all of the component fixed point supports survive, and thus the fixed points of the glued network are all unions of surviving fixed points, ultimately forming the whole set $\mathcal{M}$. While Theorem \ref{thm:rank-1-gluings} requires knowing if a local fixed point dies or survives in the glued network, survival only depends on the outgoing connections of the component subnetwork.
	
	The proofs of these results rely on a simple (and to our knowledge novel) linear algebra fact regarding determinant factorization in matrices with overlapping off-diagonal rank-1 blocks:
	\begin{lemma}\label{lemma:det-linear-algebra}
		Let 
		\begin{figure}[H]
			\centering
			\includegraphics[width=0.15\textwidth]{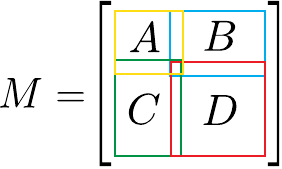}
		\end{figure}
		\noindent be an $(m+n-1)\times(m+n-1)$ matrix, consisting of the blocks $A$ ($m\times m$), $B$ ($m\times n$), $C$ ($n\times m$) and $D$ ($n\times n$),	where adjacent blocks overlap in one row/column. Note that $a_{m, m}=b_{m,1}=c_{1,m}=d_{1,1}$ is the single entry where all four matrices overlap. If $B$ or $C$ is rank 1 and $a_{m,m} \neq 0$, then $$\det M=\frac{1}{a_{m,m}} \det A \det D.$$
	\end{lemma}
	
	The rank-1 structure of the off-diagonal blocks in low-rank gluings, combined with the uniform external input $\theta$, will allows us to apply this lemma to factor the survival conditions for individual component fixed points. 
	
	These results find a more intuitive application in the realm of graph-based networks, where there exist clear survival rules for fixed points \cite{curto2019a,curto2023,curto2025a}. In particular because, for certain graph-based families, these fixed point decompositions translate into insights on the attractors of the network.

	\subsection{Application to graph-based TLNs}\label{sec:CTLN-app}
	
	 The theorems presented here extend previous findings for combinatorial threshold-linear networks (CTLNs) \cite{curto2019a,parmelee2022a} to \textit{generalized} combinatorial threshold-linear networks (gCTLNs) \cite{curto2025a}. gCTLNs are inhibition-dominated TLNs whose connectivity matrix is defined by a directed graph $G$ and neuron-specific parameters $\varepsilon_j>0$ and $\delta_j>0$:	
	\begin{equation*}
		W_{ij} = \begin{cases}
			0 & \text{if } i = j, \\
			-1 + \varepsilon_j & \text{if } j \rightarrow i \text{ in } G, \\
			-1 - \delta_j & \text{if } j \not\rightarrow i \text{ in } G.
		\end{cases}
	\end{equation*} 
	
	When these parameters are constant across neurons ($\varepsilon_j = \varepsilon, \delta_j = \delta$ for all $i$), we obtain a CTLN \cite{curto2023}. We denote the set of fixed points $\FP(W,\theta)$ of CTLNs and gCTLNs  as $\FP(G)$ and $\FP(G,\bm{\varepsilon},\bm{\delta})$, respectively.
	
	The fixed point structure of several modular CTLN architectures has previously been characterized \cite{curto2019a,parmelee2022a}. Figure \ref{fig:CTLN_results_summary} illustrates three of those: the cyclic, clique, and disjoint unions. For these cases, the global fixed point supports were previously known to be unions of local fixed point supports, as specified by Table \ref{tab:CTLN_summary} \cite{curto2019a}.
	\begin{figure}[!h]
		\begin{center}
			\includegraphics[width=0.85\textwidth]{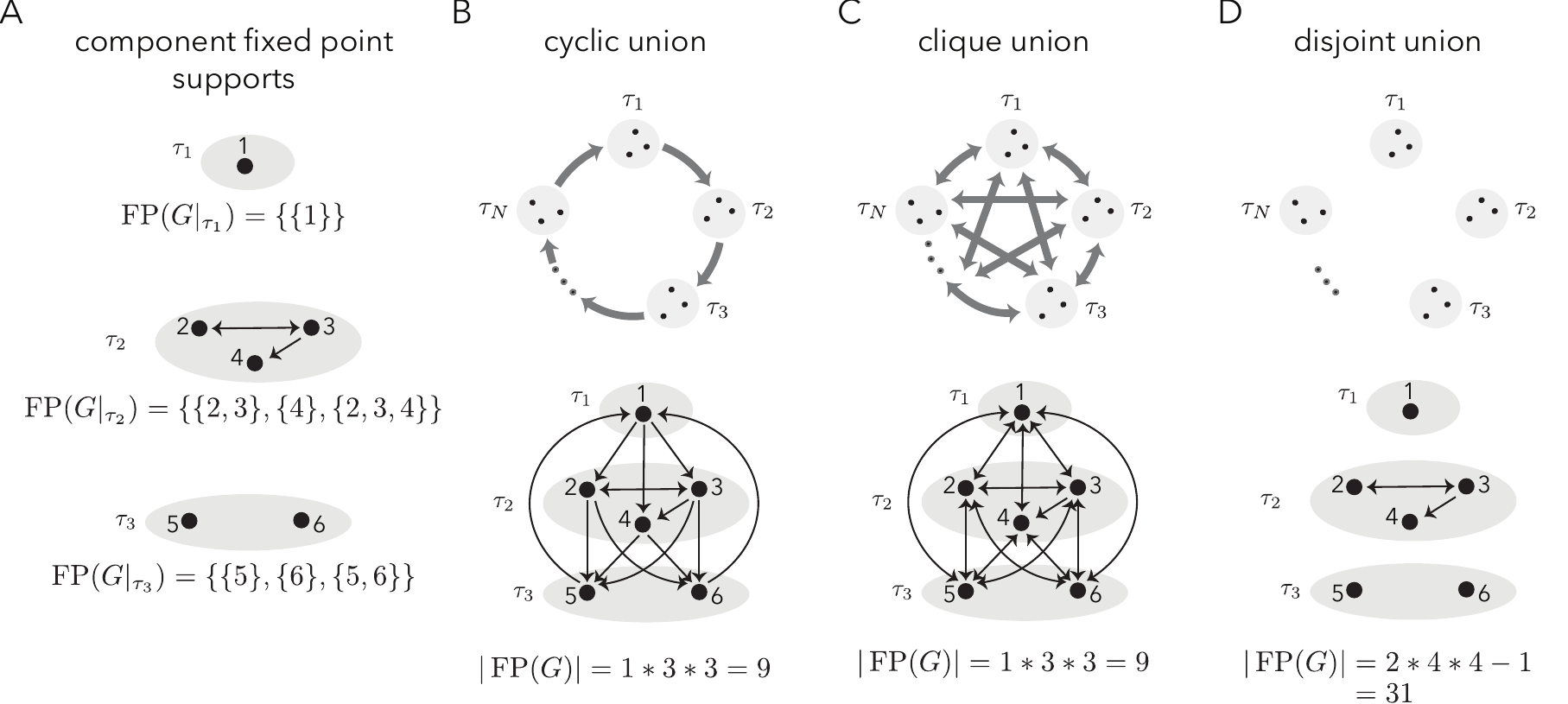}
		\end{center}
		\caption{\textbf{Illustration of CTLN theorems for three special low-rank gluings.}  
			Three specific methods for constructing glued networks from component subgraphs. 
			(A) The fixed point supports for the component subgraphs of the example in the bottom row of panels B-D.
			(B-D) A cyclic union, clique union and disjoint union, with examples. Graphs are divided into components $\tau_1, \dots, \tau_N$ with arbitrary internal connectivity, with interactions prescribed by the thick arrows indicating that every node in the source component projects to every node in the target component.
			}
		\label{fig:CTLN_results_summary}
	\end{figure}
	
	For instance, in the cyclic union of Figure \ref{fig:CTLN_results_summary}B, each global fixed point is formed as a union of exactly one support from each component (Fig. \ref{fig:CTLN_results_summary}A), resulting in $|\FP(G|_{\tau_1})| * |\FP(G|_{\tau_2})| * |\FP(G|_{\tau_3})| = 1*3 * 3 = 9$ global fixed point supports. Analogous results and counts hold for clique and disjoint unions, illustrated in Figure \ref{fig:CTLN_results_summary}C-D.
	\begin{table}[h]
		\centering
		\renewcommand{\arraystretch}{1.6}
		\begin{tabular}{|l|c|c|}
			\hline
			architecture & \textbf{fixed point decomposition theorem} & $\left|\FP(G)\right| =$ \# of fixed points \\
			\hline
			cyclic union 
			& $\FP(G) = \{\cup_i \sigma_i \mid \sigma_i \in \FP(G|_{\tau_i}) \ \forall i \in [N]\}$ 
			& $\prod_{i=1}^N |\FP(G|_{\tau_i})|$ \\
			\hline
			clique union 
			& $\FP(G) = \{\cup_i \sigma_i \mid \sigma_i \in \FP(G|_{\tau_i}) \ \forall i \in [N]\}$ 
			& $\prod_{i=1}^N |\FP(G|_{\tau_i})|$ \\
			\hline
			disjoint union 
			& $\FP(G) = \{\cup_i \sigma_i \mid \sigma_i \in\FP(G|_{\tau_i}) \cup \{\emptyset\} \ \forall i\} \setminus \{\emptyset\}$ 
			& $\prod_{i=1}^N \left( |\FP(G|_{\tau_i})| + 1 \right) - 1$ \\
			\hline
		\end{tabular}
		\caption{Summary of theorems and fixed point counts for the CTLN glued architectures of Figure \ref{fig:CTLN_results_summary}. Adapted from  \cite{curto2023}.}
		\label{tab:CTLN_summary}
	\end{table}
	
	By applying our new Theorems \ref{thm:low-rank-gluing-menu} and \ref{thm:rank-1-gluings}, we can now prove that the CTLN ``gluing rules'' from Table \ref{tab:CTLN_summary} extend to the more flexible gCTLN family, which allows each neuron to have its own $\varepsilon,\delta$ parameters:
	\begin{theorem}[generalized cyclic union]\label{thm:generalized-cyclic-union}
		Let $G$ be a cyclic union of components $\tau_1, \dots, \tau_N$, and let $(G, \bm{\varepsilon}, \bm{\delta})$ be a gCTLN. Then
		$$ \sigma \in \FP(G, \bm{\varepsilon}, \bm{\delta}) \quad \Leftrightarrow \quad \sigma = \bigcup_{\ell=1}^{N} \sigma_\ell, \text{ with } \sigma_\ell \in \FP(G|_{\tau_\ell}, \bm{\varepsilon}_{\tau_\ell}, \bm{\delta}_{\tau_\ell}).$$
	\end{theorem}
	
	\begin{theorem}[generalized clique and disjoint union]\label{thm:generalized-disjoint-clique-union}
		Let $G$ be a graph with a partition of its vertices $\{\tau_1 \mid \cdots \mid \tau_N\}$, and let $(G, \bm{\varepsilon}, \bm{\delta})$ define a gCTLN.
		\begin{enumerate}
			\item If $G$ is a clique union of $G|_{\tau_1}, \ldots, G|_{\tau_N}$, then
			$$ \sigma \in \FP(G, \bm{\varepsilon}, \bm{\delta}) \quad \Leftrightarrow \quad \sigma = \bigcup_{\ell=1}^{N} \sigma_\ell, \text{ with } \sigma_\ell \in \FP(G|_{\tau_\ell}, \bm{\varepsilon}_{\tau_\ell}, \bm{\delta}_{\tau_\ell}).$$
			\item If $G$ is a disjoint union of $G|_{\tau_1}, \ldots, G|_{\tau_N}$, then 
			$$\sigma \in \FP(G, \bm{\varepsilon}, \bm{\delta}) \quad \Leftrightarrow \quad \sigma = \bigcup_{\ell=1}^{N} \sigma_\ell, \text{ with }  \sigma_\ell \in \FP(G|_{\tau_\ell}, \bm{\varepsilon}_{\tau_\ell}, \bm{\delta}_{\tau_\ell}) \cup \{\emptyset\}$$
		\end{enumerate}
	\end{theorem}
	
	Because the neuron-specific parameters in a gCTLN are vectors $\bm{\varepsilon},\bm{\delta} \in \mathbb{R}^n$, we use $\bm{\varepsilon}_{\tau_\ell}$ and $\bm{\delta}_{\tau_\ell}$ to denote the restriction of these to the indices in $\tau_\ell$. This ensures that the parameters for individual neurons within the subnetworks match their corresponding values in the full graph $G$.
	
	While these fixed point decomposition theorems are identical for both CTLNs and gCTLNs, this extension of the gluing rules to gCTLNs required the development of new theoretical tools, such as Theorems \ref{thm:low-rank-gluing-menu} and \ref{thm:rank-1-gluings}. Furthermore, even when the fixed point supports are identical for a CTLN and a gCTLN built from the same graph (which is not always the case), important differences can arise in the dynamics.
	
	\paragraph{Dynamics of graph-based low-rank gluings.}\label{sec:fp-correspondance}
	
	In CTLNs, there is often a correspondence between attractors and the subset of minimal (by inclusion) fixed point supports $\FP_{\min}(G)$ \cite{parmelee2022}. Here we show that gCTLNs also follow this heuristic. We illustrate this correspondence using the cyclic union network from Figure \ref{fig:CTLN_results_summary}A, reproduced in Figure \ref{fig:attractor_correspondance_cyclic_union}A. There, the minimal local supports (left, in bold) combine to form the minimal global supports (right, in bold). 
		\begin{figure}[!h]
		\begin{center}
			\includegraphics[width=\textwidth]{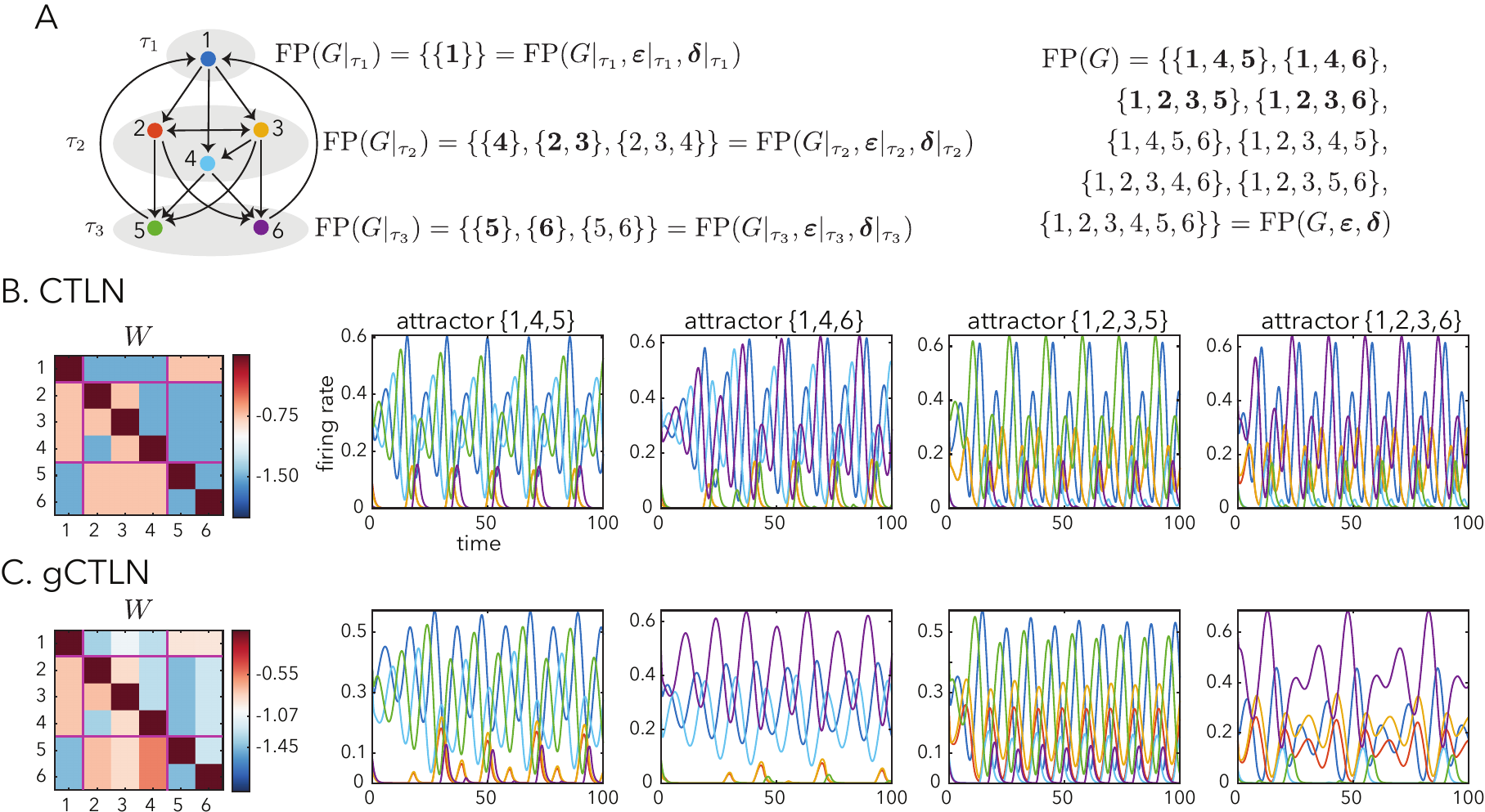}
		\end{center}
		\caption{\textbf{Attractor correspondence in CTLNs and gCTLNs.} 
			(A) A cyclic union graph is shown with its set of local and global fixed point supports for the CTLN and gCTLN matrices of panels B-C. For the specific choice of gCTLN parameters of panel C, the CTLN fixed point supports match those of the gCTLN. The minimal supports, which correspond to attractors, are in bold. 
			(B) The four attractors corresponding to the minimal fixed points of the CTLN, labeled by which fixed point perturbation gave rise to it.
			(C) Same as B, but for a gCTLN.
		}
		\label{fig:attractor_correspondance_cyclic_union}
	\end{figure}
	
	By choosing initial conditions that are small perturbations of the (unstable) fixed points associated with these supports, the network quickly settles into trajectories that computationally appear to be stable limit cycles (Fig. \ref{fig:attractor_correspondance_cyclic_union}B-C). This holds for both CTLNs and gCTLNs built from the graph of panel A. Each attractor exhibits sequential activity cycling through the components, with the highest-firing neurons matching the corresponding supports. We have thus obtained $|\FP_{\min}(G|_{\tau_1})| * |\FP_{\min}(G|_{\tau_2})| * |\FP_{\min}(G|_{\tau_3})| = 1 * 2 * 2 = 4$ distinct sequential patterns from this 6-node network, for both the CTLN and gCTLN. 
	
	\paragraph{Example application to locomotion tuning.} Sequential attractors are ideal for modeling rhythmic processes, like locomotion. In prior work, we used CTLN cyclic unions to model a quadruped's bound gait \cite{londonoalvarez2026}, where the front limbs synchronously activate, followed by the back limbs (Fig. \ref{fig:bound_gait}A).  Presumably, the stronger muscles in the back limbs should be driven by a stronger neural signal (i.e., a higher firing rate) than that driving front limbs. However, this cannot be achieved with CTLNs, since the binary inhibition parameters force symmetry in the firing rates.
	
	With gCTLNs we can easily break this symmetry, as illustrated in Figure \ref{fig:bound_gait}. There, both the CTLN and gCTLN are built on the same defining graph (Fig. \ref{fig:bound_gait}B) and, as Theorem \ref{thm:generalized-cyclic-union} shows, they both have the same set of fixed point supports (Fig. \ref{fig:bound_gait}C). However, the resulting CTLN attractor (Fig. \ref{fig:bound_gait}D) shows symmetric firing rates for the front and back limbs, while the gCTLN attractor (Fig. \ref{fig:bound_gait}E) can be easily tuned to produce a stronger activation for the back limbs.
	
	\begin{figure}[!h]
		\begin{center}
			\includegraphics[width=0.95\textwidth]{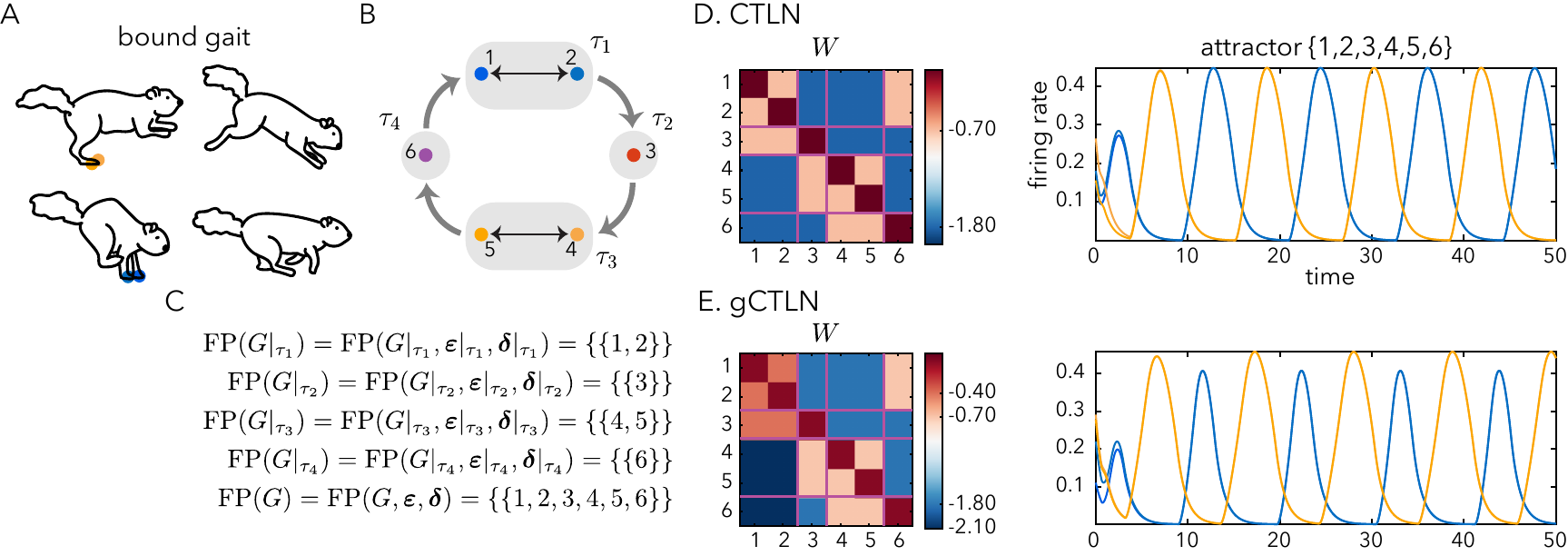}
		\end{center}
		\caption{\textbf{Bound gait as a CTLN and a gCTLN.}  
			(A) The bound gait. Drawing adapted from \cite{vanderweele2001}.
			(B) A cyclic union designed to mimic the limb activations of the bound gait. $\tau_1$ (front limbs) and $\tau_3$ (back limbs) contain the readout nodes.
			(C) Fixed point supports for the network in B, identical for both the CTLN (panel D) and gCTLN (panel E).
			(D-E) Connectivity matrix for the CTLN and gCTLN, respectively, and  the attractor corresponding to the $\{1,2,3,4,5,6\}$ support. Only readout nodes are shown. Adapted from \cite{londonoalvarez2026}.
		}
		\label{fig:bound_gait}
	\end{figure}

	More generally, this correspondence between minimal fixed points and attractors provides a recipe for engineering compositional dynamics. By connecting the same motifs in different ways, we can produce fundamentally different global behaviors, as illustrated in Figure \ref{fig:CTLN_vs_gCTLN_dynamics}. First, note that by Theorems \ref{thm:generalized-cyclic-union} and \ref{thm:generalized-disjoint-clique-union}, because each isolated component possesses only a single fixed point of full support, both the cyclic union and the clique union have exactly one global fixed point. Regarding the dynamics, though these architectures use identical building blocks, and even share the same set of fixed point supports in the cyclic and clique union cases, they exhibit entirely different global dynamics.
		\begin{figure}[!h]
		\begin{center}
			\includegraphics[width=0.95\textwidth]{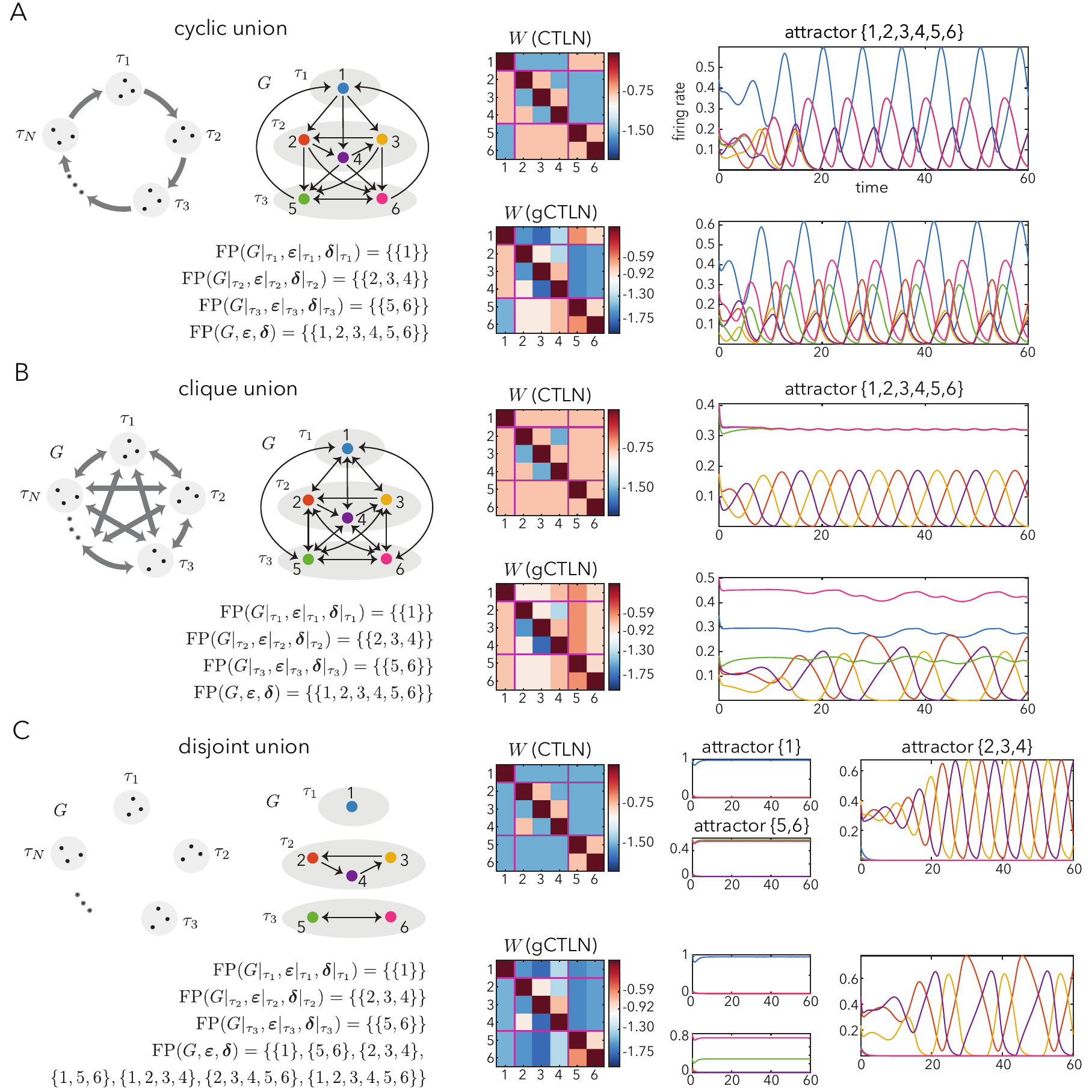}
		\end{center}
		\caption{\textbf{Fixed points and dynamics of CTLN vs. gCTLN for three specific graph constructions.}
			Comparison of attractors for CTLNs (top traces) and gCTLNs (bottom traces) across three architectures built from the same components. 
			(A) Sequential dynamics arising from a cyclic union. 
			(B) A fusion attractor arising from a clique union 
			(C) Component-wise winner-take-all dynamics arising from a disjoint union.}
		\label{fig:CTLN_vs_gCTLN_dynamics}
	\end{figure}  
	
	Specifically, we observe sequential attractors from cyclic unions (Fig. \ref{fig:CTLN_vs_gCTLN_dynamics}A); ``fusion'' attractors where component attractors are co-expressed in clique unions (Fig. \ref{fig:CTLN_vs_gCTLN_dynamics}B); and component-wise winner-take-all dynamics in disjoint unions (Fig. \ref{fig:CTLN_vs_gCTLN_dynamics}C). In particular, for cyclic and clique unions, this capacity to reliably generate a combinatorially large set of compositional sequential or fusion attractors might be useful for modeling neural mechanisms such as feature binding, where distinct features must be dynamically bound into a single coherent representation \cite{treisman1996,roskies1999a}.
	
	Across all three architectures, the transition from CTLN parameters to gCTLN parameters induces a symmetry-breaking in the firing rate curves. For example, for all glued CTLNs, nodes 5 and 6 are perfectly synchronized, and the periodic cycling through nodes 2, 3, and 4 is perfectly symmetric. By contrast, in the corresponding gCTLNs, the synchronization is lost and the sequential dynamics become noticeably irregular. Nevertheless, the underlying qualitative nature of the attractors remains predictable from the glued network architecture. A similar phenomenon, where qualitative dynamics are preserved despite quantitative variation, has been observed in various biological systems \cite{long2008,tang2010}.	
	
	\paragraph{Graph rules in gCTLNs.}
	
	Graph rules refer to a set of results relating graph structure to their set of fixed point supports in CTLNs \cite{curto2019a,curto2023}. In addition to generalizing fixed point decomposition results for modular architectures to gCTLNs, we can also apply Theorems \ref{thm:generalized-cyclic-union} and \ref{thm:generalized-disjoint-clique-union} to graphs where each component $\tau_i$ is just a single node to recover and extend graph rules from CTLNs to gCTLNs. 
	
		\begin{figure}[!h]
		\centering
		\includegraphics[width=0.5\textwidth]{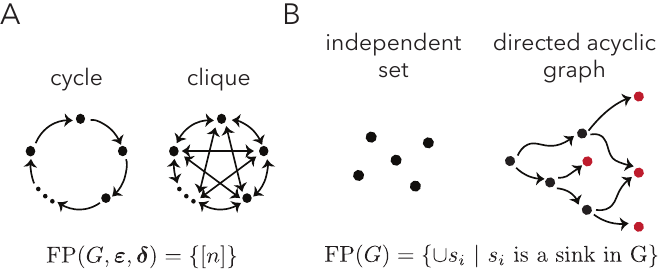}
		\caption{\textbf{gCTLNs graph rules.}
			(A) Cycles and cliques yield unique full-support fixed points.
			(B) The fixed point supports for independent sets and directed acyclic graphs are sinks and unions of sinks.
			Adapted from \cite{curto2023}.
		}
		\label{fig:grap_rules_applications}
	\end{figure}
	
	A cyclic and clique union on single-node components define a cycle and a clique (Fig. \ref{fig:grap_rules_applications}A). A disjoint union on single-node components gives an independent set (Fig. \ref{fig:grap_rules_applications}B, left). A directed acyclic graph (DAG) is a directed graph with no directed cycles (Fig. \ref{fig:grap_rules_applications}B, right).	With these definitions, we obtain the following result:
	\begin{corollary}\label{cor:graph-rules}
		Let $G$ be a graph on $n$ nodes, and $\sigma \subseteq [n]$.
		\begin{enumerate}
			\item if $G$ is a cycle, then $ \FP(G, \bm{\varepsilon}, \bm{\delta}) = \{[n]\}$.
			\item If $G$ is a clique, then $ \FP(G, \bm{\varepsilon}, \bm{\delta}) = \{[n]\}$.
			\item If $G$ is an independent set, then $\FP(G, \bm{\varepsilon}, \bm{\delta}) = \left\{ \bigcup_{i \in I} \{i\} \,\middle|\, \emptyset \neq I \subseteq [n] \right\}$, the set of all $2^{n}-1$ unions of sinks.
			\item If $G$ is a directed acyclic graph with sinks $s_1, \ldots, s_{\ell}$, then $\FP(G, \bm{\varepsilon}, \bm{\delta}) = \left\{ \bigcup_{i \in I} \{s_i\} \,\middle|\, \emptyset \neq I \subseteq [\ell] \right\}$, the set of all $2^{\ell}-1$ unions of sinks.
		\end{enumerate}
	\end{corollary}
	
	To obtain the first three, it suffices to note that a single node is always a fixed point of its own subnetwork, and then to apply Theorems \ref{thm:generalized-cyclic-union} and \ref{thm:generalized-disjoint-clique-union} to single-node components. Corollary \ref{cor:graph-rules}(d) follows from first applying Theorem \ref{thm:removal-of-dominated-nodes} to remove all non-sink nodes, which reduces the graph to an independent set of its sinks, at which point Corollary \ref{cor:graph-rules}(c) applies.
	
	In CTLNs these graph rules arose as special cases of a result that we are not able to prove here: the uniform in-degree theorem.  We say that $G|_\sigma$ has uniform in-degree $d$ if every node $i \in \sigma$ has $d$ incoming edges from within $G|_\sigma$ (an independent set has uniform in-degree $d=0$, a cycle has uniform in-degree $d=1$, and an $n$-clique is uniform in-degree with $d=n-1$). For these graphs we have the following result:
	\begin{theorem}[uniform in-degree \cite{curto2019a}]\label{thm:uniform-in-deg}
		Let $(G,\varepsilon,\delta)$ be a CTLN and let $G|_{\sigma}$ be an induced subgraph of $G$ with uniform in-degree $d$. For $k \notin \sigma$, let $d_k$ denote the number of edges $i \to k$ for $i \in \sigma$. Then $\sigma \in \FP(G|_{\sigma})$, and
		\[
		\sigma \in \FP(G|_{\sigma \cup k}) \Leftrightarrow d_k \le d.
		\]
		In particular, $\sigma \in \FP(G)$ if and only if there does not exist $k \notin \sigma$ such that $d_k > d$.
	\end{theorem}
	\begin{figure}[!h]
		\centering
		\includegraphics[width=0.8\textwidth]{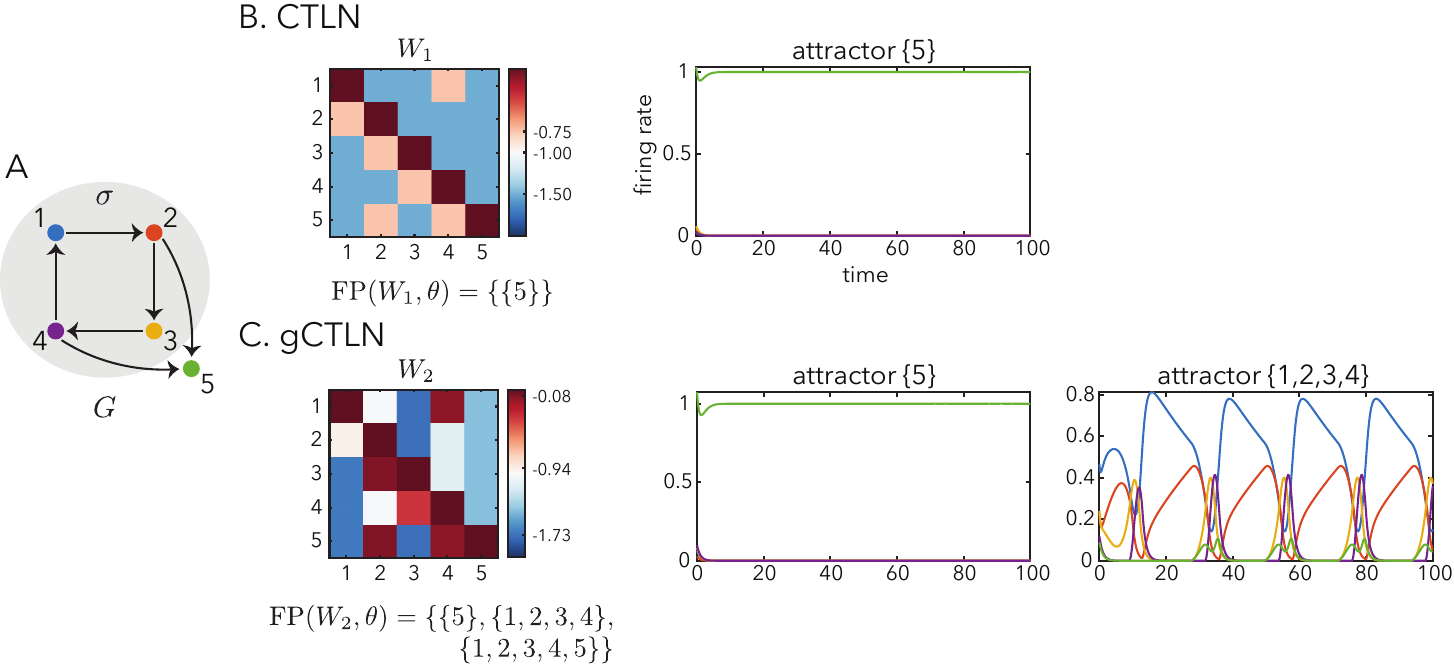}
		\caption{\textbf{A CTLN and a gCTLN defined on a graph satisfying the assumptions of Theorem \ref{thm:uniform-in-deg}.} Here, $\sigma = \{1,2,3,4\}$ and $k = 5$.
			(A) Graph $G$ defining both networks. $\sigma$ has uniform in-degree $d=1$, and $k=5$ has in-degree $d_5=2$.
			(B) CTLN connectivity matrix and its fixed points. Since $2 = d_5 > d = 1$, Theorem \ref{thm:uniform-in-deg} gives $\{1,2,3,4\} \notin \FP(W_1,b)$. Computationally, one attractor was found, corresponding to $\{5\} \in \FP(W_1,\theta)$. 
			(C) gCTLN connectivity matrix and its fixed points. Note that $\{1,2,3,4\} \in \FP(W_1,\theta)$, even though $2 = d_5 > d = 1$ in $G$. Computationally, two attractors were found, corresponding to $\{5\}, \{1,2,3,4\} \in \FP(W_1,\theta)$.
		}
		\label{fig:uniform_in_degree_example}
	\end{figure}
	
	The proof of this CTLN result relied on the fact that all synaptic weights were identical, which implied that the input to a node was determined by simple edge counts (row sums). This is no longer true in gCTLNs, as the neuron-specific parameters $\varepsilon_j$ and $\delta_j$ make the row sums too heterogeneous.
	
	This can be clearly seen in Figure \ref{fig:uniform_in_degree_example}, where the induced subgraph $G|_{\{1,2,3,4\}}$ satisfies the uniform in-degree condition to die in the CTLN ($2 = d_5 > d = 1$), but still survives in the gCTLN. This is because by choosing $\varepsilon_i,\delta_i$ appropriately, the total intra-component connectivity can be made stronger than the total output to the external node 5.
	
	While Corollary \ref{cor:graph-rules} confirms that independent sets, cliques, and cycles all yield full-support fixed points in isolation, it does not say anything precise about their survival when embedded into a larger network, which is the content of Theorem \ref{thm:uniform-in-deg}. Intuitively, we know from Lemma \ref{lemma:single-node-survival} that these motifs will survive as long as they don't turn on any targets, and will die otherwise. However, further work is needed to obtain formal survival rules analogous to those of Theorem \ref{thm:uniform-in-deg} for gCTLNs.
	
	\subsection{Application to sensory network modeling}\label{sec:mollusk-app}
	
	In this section, we provide an example application of Theorem \ref{thm:generalized-cyclic-union} by modeling the sensory network of a blind marine mollusk, called \textit{Clione}, inspired by the model of Varona et al. \cite{varona2002}. Because Clione is blind, when it chemically senses its prey, it must explore its surroundings by swimming in complex patterns until it finds it. The direction of swimming is controlled by the tail, which receives input from the statocyst -- a gravitational sensory organ. The statocyst in turn is covered by mechanoreceptors that receive inputs from the statolith, a stone-like structure moving inside the statocyst under the effect of gravity  (Fig. \ref{fig:clione-transitions}A, left) \cite{panchin1995,varona2002}. 		
		\begin{figure}[h!]
		\begin{center}
			\includegraphics[width=\textwidth]{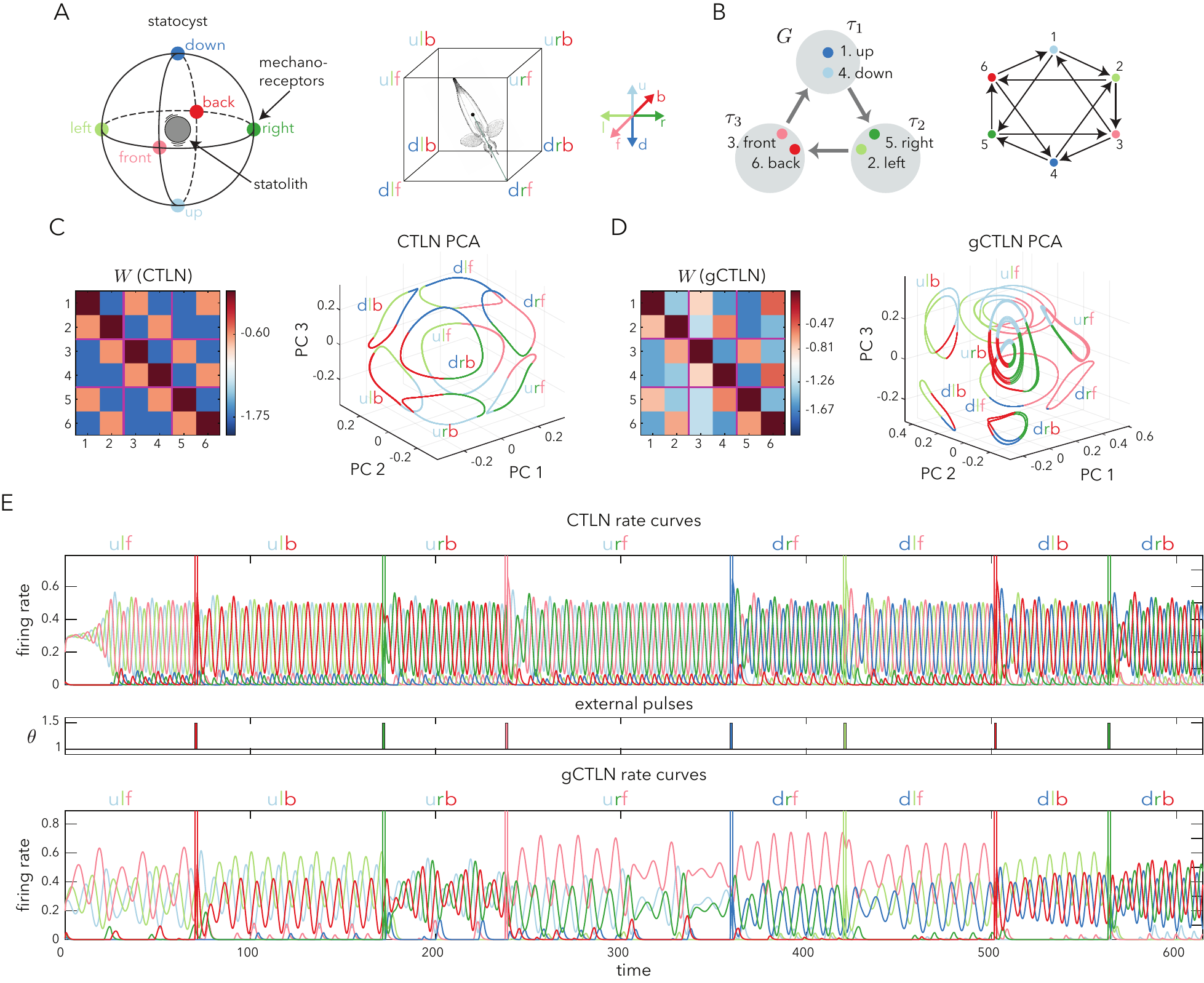}
		\end{center}
		\caption{\textbf{Clione's sensory network.}
			(A) Left: diagram of Clione's statocyst, adapted from \cite{varona2002}. Center and right: all possible combinations of up/down, right/left, and front/back swimming directions, alongside their corresponding color-coded vectors (e.g., ulf: up-left-front; ulb: up-left-back).
			(B) Left: a cyclic union of three independent sets. Center: an equivalent representation of the graph in the left panel, with all edges explicitly drawn. 
			(C) A connectivity matrix defining a CTLN ($\varepsilon = 0.4, \delta = 0.75$), and a principal component analysis (PCA) projection of the 8 resulting attractors.
			(D) A connectivity matrix defining a gCTLN (for exact parameter values, see the accompanying code), and a PCA projection of the 8 resulting attractors.
			(E) Rate curves for the attractors in panels C and D, visited sequentially via external $\theta$-pulses (colored by target neuron) triggering transitions between swimming patterns. Top labels indicate the active combination. Note that directions consecutively change one at the time, until all attractors are visited once.
		}
		\label{fig:clione-transitions}
	\end{figure} 
	
	In \cite{varona2002}, the authors modeled Clione's sensory network using a six-neuron network of generalized Lotka-Volterra units with mutually inhibitory non-symmetric connectivity. Each unit represents a receptor neuron in the statocyst. Their architecture leads to a state of winnerless competition in which neurons take turns to fire, chaotically. These chaotic trajectories are used by the authors to explain the complex, seemingly random changes in swimming direction observed during hunting.
	
	Similarly, we propose a six-unit CTLN to model this system. While CTLNs, by definition, also feature mutually inhibitory non-symmetric connectivity, the single unit dynamics are perceptron-like, as opposed to Lotka-Volterra. Since cyclic unions are particularly well-suited to model pattern generation, and because the network consists of neurons representing receptors that ultimately control the direction of swimming (up/down, right/left, front/back), we construct the network as a cyclic union of three components, where each component consists of an independent set of two neurons representing opposing directions (Fig. \ref{fig:clione-transitions}B, left). 
	
	This architecture forces a component-wise choice between opposing directions, and by Theorem \ref{thm:generalized-cyclic-union}, it will compositionally represent all possible swimming directions in 3D space. Indeed, computing the minimal fixed point supports for such a network is straightforward for both a CTLN and a gCTLN. By Theorem \ref{thm:generalized-cyclic-union} and Corollary \ref{cor:graph-rules}(c), each minimal support is formed by selecting exactly one node from each component, resulting in the $2^3=8$ minimal fixed point supports:
	\begin{equation}\label{eq:octahedron-FPmin}
		\FP_\text{min}(G, \bm{\varepsilon}, \bm{\delta}) = \{\{1,2,3\},\{1,2,6\},\{1,3,5\},\{1,5,6\},\{2,3,4\},\{2,4,6\},\{3,4,5\},\{4,5,6\}\}.
	\end{equation} By the number-direction correspondence of Figure \ref{fig:clione-transitions}B, these supports map directly to the eight octants labeled in the cube of Figure \ref{fig:clione-transitions}A.

	For the CTLN case, numerical simulations confirm that each minimal fixed point support in Equation \eqref{eq:octahedron-FPmin} corresponds to an attractor (Fig. \ref{fig:clione-transitions}C). While the principal component analysis (PCA) projections show limit cycles, quasiperiodic trajectories can also be obtained by varying the $\varepsilon$ and $\delta$ parameters \cite{londonoalvarez2024}. Furthermore, each attractor is easily accessible via initial conditions or external $\theta$-pulses sent to the target direction. For instance, in the top panel of Figure \ref{fig:clione-transitions}E, the network is initialized in the attractor corresponding to up-left-front (ulf), then a transient pulse sent to neuron 6 (back) arrives, causing the system to converge to the up-left-\textit{back} (ulb) attractor. Transitions are possible between any two attractors, as demonstrated by the sequential activation of all eight attractors in Figure \ref{fig:clione-transitions}E. 
		
	Because the binary synapses of the CTLN enforce symmetry in the connectivity, the resulting attractors and their basins are also perfectly symmetric:
	\begin{theorem}\label{thm:octahedral-basins}
		Let the graph $G$ of Figure \ref{fig:clione-transitions}B define a CTLN. Suppose that for this CTLN there exists an attractor $\mathcal{A}_{\sigma}$ corresponding to one of the minimal fixed point supports of Eq. \eqref{eq:octahedron-FPmin} and that there are no attractors that do not correspond to a minimal fixed point support. Then there are exactly 8 attractors, one for each minimal fixed point support, and their respective basins of attraction $\mathcal{B}_{\sigma}$ are contained in the sets
		\begin{equation}\label{eq:octahedral-basins}
			C_{ijk} = \{(x_1,x_2,x_3,x_4,x_5,x_6)  \in \RR^6 \, | \, x_i \leq x_{\pi(i)} ,\, x_j \leq x_{\pi(j)},\, x_k \leq x_{\pi(k)}\},
		\end{equation}  
		where $\pi = (14)(25)(36)$.
	\end{theorem}
	\begin{proof}
		Note that the network is symmetric under the following permutations in $S_6$:
		\begin{align*}
			\pi_1 &= (14) \\
			\pi_2 &= (25) \\
			\pi_3 &= (36) 
		\end{align*} so that $\pi = \pi_1 \pi_2 \pi_3$. Therefore the group of symmetries of the network is $\langle \pi_1, \pi_2, \pi_3 \rangle  \cong \ZZ_2 \times \ZZ_2 \times \ZZ_2$.  As a consequence of these symmetries, the existence of an attractor $\mathcal{A}_{\sigma}$ corresponding to a minimal fixed point support $\sigma \in \FP_\text{min}(G)$ implies the existence of all other attractors corresponding to the rest of the minimal fixed point supports in $\FP_\text{min}(G)$. Since there are no other attractors by assumption, the state space must be partitioned into exactly eight symmetric basins of attraction $\mathcal{B}_{\sigma}$. 
		
		On the other hand, since $\ell$ and $\pi(\ell)$ receive input from the same neurons (for $\ell = 1,\dots,6$), it is impossible to cross the hyperplanes $x_\ell = x_{\pi(\ell)}$, as we show next. Suppose that $x_\ell(t_0) = x_{\pi(\ell)}(t_0)$ for some $t_0$. The dynamic equations for neurons $\ell$ and $\pi(\ell)$ are
		\begin{align*}
			\frac{dx_\ell}{dt} &= -x_\ell + \left[\sum_{j \neq \ell,\pi(\ell)} W_{ij}x_j + (-1-\delta)x_{\pi(\ell)} + \theta \right]_+ \\
			\frac{dx_{\pi(\ell)}}{dt} &= -x_{\pi(\ell)} + \left[\sum_{j \neq \ell,\pi(\ell)} W_{ij}x_j + (-1-\delta)x_{\ell} + \theta \right]_+.
		\end{align*} Both summation terms are the same, since $\ell$ and $\pi(\ell)$ receive input from the same neurons. Since $x_\ell(t_0) = x_{\pi(\ell)}(t_0)$ by assumption, we have that $\frac{dx_\ell}{dt} = \frac{dx_{\pi(\ell)}}{dt}$, and therefore $x_\ell(t) = x_{\pi(\ell)}(t)$ for all $t\geq t_0$. This proves that it is impossible to cross the hyperplanes  $x_\ell = x_{\pi(\ell)}$. This implies that the sets $$C_{ijk} = \{(x_1,x_2,x_3,x_4,x_5,x_6)  \in \RR^6 \, | \, x_i \leq x_{\pi(i)} ,\, x_j \leq x_{\pi(j)},\, x_k \leq x_{\pi(k)}\}$$ are forward-invariant and also partition the state space into eight sets. Since the basins of attraction $\mathcal{B}_{\sigma}$ are also forward-invariant and partition the state space into eight sets, we get the desired result.
	\end{proof}
	
	Because the basins of attraction are divided by the hyperplanes $ x_i = x_{\pi(i)}$, an external input has to be strong or long enough, or either timed precisely when the trajectory is near this boundary, to push it across into a different basin. Therefore, the success of each transition depends on a combination of internal and external factors. This dependency introduces a natural source of variability, as multiple conditions must align for Clione to effectively change directions \cite{londonoalvarez2024}. This mechanism provides a complementary explanation for the observed randomness in swimming behavior to that proposed by Varona et al. \cite{varona2002}.
			
	It is natural to ask if this construction depends on the strict binary synapses of CTLNs. While the theoretical results above guarantee identical fixed points for both the CTLN and gCTLN cases, the preservation of attractors and the ability to transition between them for gCTLNs is not guaranteed. For this, numerical simulations confirm the attractors are preserved under arbitrary choices of gCTLN parameters (Fig. \ref{fig:clione-transitions}D), and that all eight attractors remain accessible via sequential external pulses (Fig. \ref{fig:clione-transitions}E, bottom panel). However, the ability to transition between attractors is less robust than in the CTLN case, since Theorem \ref{thm:octahedral-basins} does not hold for gCTLNs. An investigation of these limits is postponed to future work. Nonetheless, Figure \ref{fig:clione-transitions} demonstrates the robust nature of the architecture to the heterogeneous synapses of gCTLNs. Such robustness could be interpreted as reflecting animal-to-animal variability, where the existence of the attractors and the sequential activation patterns are preserved across animals, while the specific firing rates may vary. 
	
	Ultimately, our model offers a complementary perspective to the work of Varona et al. \cite{varona2002}. Remarkably, by engineering our network to support the required fixed points via compositionality, we arrive at the same underlying adjacency matrix used in their model (Fig. \ref{fig:clione-transitions}B, right). However, the equivalent cyclic union representation more explicitly highlights the structural origins of the dynamics: because the components are defined as independent sets, there is a component-wise competition leading to local winner-take-all dynamics. At the same time, the directionality of the cyclic union architecture makes evident the origin of the sequential dynamics.
		
	\section{Proofs}\label{sec:proofs}
	
	This section provides proofs for Theorems \ref{thm:low-rank-gluing-menu} to \ref{thm:generalized-disjoint-clique-union}. We begin by reviewing the necessary technical background, originally introduced in \cite{curto2019a}. This background will be enough to prove Theorems \ref{thm:low-rank-gluing-menu} and \ref{thm:rank-1-gluings}. To prove Theorems \ref{thm:generalized-cyclic-union} and \ref{thm:generalized-disjoint-clique-union}, we will need, in addition to Theorems \ref{thm:low-rank-gluing-menu} and \ref{thm:rank-1-gluings}, the development of new theory based on a graph property called domination. Domination was introduced in \cite{curto2025a}, and here we develop the necessary consequences to be able to apply Theorems \ref{thm:low-rank-gluing-menu} and \ref{thm:rank-1-gluings} to gCTLNs.

	\subsection{Technical background}\label{sec:technical-background}
	In this section, we review essential background on TLNs and domination and establish notation, with the goal of making the paper self-contained. First, we introduce tools that help us characterize fixed point supports in TLNs, following \cite{curto2019a}. Then, we present background on domination, a graph property with important consequences for the dynamics, following \cite{curto2025a}.
	
	\subsubsection{TLN background}
		
	We begin by formally reviewing some general background on threshold-linear networks and introducing the technical tools that will allow us to prove that a given subset of nodes supports a fixed point.
	
	Formally, a \textit{threshold-linear network} $(W,\bm{b})$ on $n$ neurons is a system of ordinary differential equations 
	\begin{equation}\label{eq:TLNdynamics}
		\tau_i\dfrac{dx_i}{dt} = -x_i + \left[\sum_{j=1}^n W_{ij}x_j+b_i \right]_+, \quad i = 1,\ldots,n,
	\end{equation}
	where $[W_{ij}] = W \in \RR^{n\times n}$ and $[b_i]  = \bm{b} \in \RR^{n}$.
	
	The results we prove here only hold for TLNs with \emph{uniform} external input $b_i = \theta$ for all $i \in [n]$. In this case, we abuse notation and denote the TLN $(W,\theta\mathbbm{1})$ by just $(W,\theta)$. Additionally, we assume $\tau_i=1$ for all $i \in [n]$. Moreover, we assume all networks to be \emph{competitive} \emph{nondegenerate} TLNs. We say that a TLN $(W,\bm{b})$ is {\it competitive} if $W_{ij}\leq0$, $W_{ii} = 0$ and $b_i\geq 0$ for all $i,j\in[n]$. 
	
	The definition of non-degeneracy requires some Cramer's determinants of certain submatrices to not vanish. To this end, by $((A_\sigma)_i;\bm{b}_\sigma)$ we denote the matrix $A_\sigma$ where the column corresponding to the index $i \in \sigma$ has been replaced by $\bm{b}_\sigma$, with the subindex denoting restriction to the rows/columns given by $\sigma$. We denote by $[n]$ the set of indices $\{1,\dots,n\}$.
	
	\begin{definition}\label{def:nondegenerate}
		We say that a TLN $(W,\bm{b})$ is {\it nondegenerate} if 
		\begin{enumerate}
			\item $b_i>0$ for at least one $i \in [n]$
			\item $\det(I-W_\sigma) \neq 0$ for each $\sigma \subseteq [n]$, and 
			\item for each $\sigma \subseteq [n]$ such that $b_i>0$ for all $i \in \sigma$,  the corresponding Cramer's determinant is nonzero: $\det((I-W_\sigma)_i;\bm{b}_\sigma) \neq 0$. 
		\end{enumerate}
	\end{definition}
	
	A fixed point $x^*$ of the system in Eq. \eqref{eq:TLNdynamics} is a solution to $\frac{dx}{dt} = 0$, that is, where $x^*_i = \left[\sum_{j=1}^n W_{ij}x_j+b_i \right]_+$. For competitive nondegenerate TLNs, we can unambiguously label all fixed points by their support, $\supp{x^*} = \{i \in [n] \mid x^*_i>0\}$ \cite{curto2019a}. The set of all fixed point supports for a given TLN is:
	\begin{equation}\label{eq:FP-TLN-general}
		\FP(W,\bm{b}) \od \{\sigma \subseteq [n] ~|~ \sigma \text{ is the support of a fixed point } x^* \}.
	\end{equation}
	
	A nice algebraic tool to determine if a given support $\sigma$ belongs to $\FP(W, \bm{b})$ is Theorem \ref{thm:sgn-conditions} below. It provides a simple test based on the signs of specific determinants. These determinants, defined below, are important combinatorial objects associated with a given TLN, and have interesting connections to the theory of chirotopes \cite{londonoalvarez2024}. For any $\sigma \subseteq [n],$ we define
	\begin{equation}\label{eq:s_i}
		s_i^\sigma \od \det((I-W_{\sigma\cup\{i\}})_i;b_{\sigma\cup\{i\}}), \;\; \text{for} \;\; i \in [n].
	\end{equation} 
	
	This brings us to the full characterization of $\FP(W,\bm{b})$ in terms of the signs of the $s_i^\sigma$. This theorem, illustrated in Figure \ref{fig:sgn-conditions}, is the primary workhorse of our subsequent proofs. 
	\begin{theorem}[sign conditions, \cite{curto2019a}]\label{thm:sgn-conditions}
		Let $(W,\bm{b})$ be a TLN on $n$ neurons. For any nonempty $\sigma \subseteq [n]$,
		$$\sigma \in \FP(W_\sigma,\bm{b}_\sigma)  \;\; \Leftrightarrow \;\; \sgn s_i^\sigma = \sgn s_j^\sigma
		\text{ for all } i,j \in \sigma.$$
		In that case $$\sgn s_i^\sigma = \sgn\det(I-W_\sigma) \od \idx(\sigma) \text{ for all } i \in \sigma. $$ and we say that $\sigma$ is permitted (and forbidden otherwise). Furthermore,
		$$\sigma \in \FP(W,\bm{b})  \;\; \Leftrightarrow \;\; \sgn s_i^\sigma = \sgn s_j^\sigma = -\sgn s_k^\sigma
		\text{ for all } i,j \in \sigma,\; k \not\in \sigma,$$
	\end{theorem} 

	\begin{figure}[!h]
		\begin{center}
			\includegraphics[width=\textwidth]{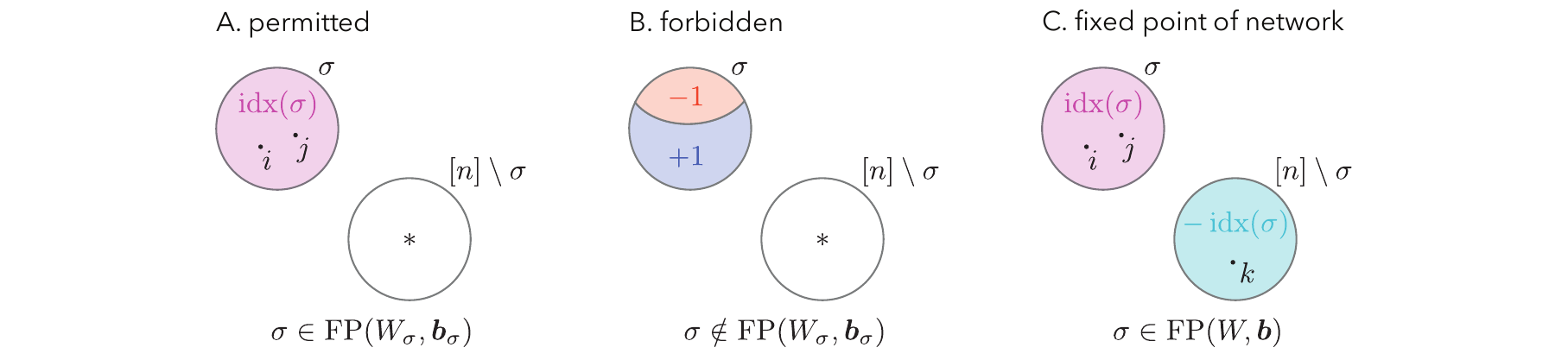}
		\end{center}
		\caption{\textbf{Theorem \ref{thm:sgn-conditions} (sign conditions, \cite{curto2019a}).} Same color indicates the same $\sgn s_i^\sigma$ for $i$ in the given component.
			(A) A subset $\sigma$ is permitted if and only if $\sgn s_i^\sigma = \sgn s_j^\sigma = \idx(\sigma) \text{ for all } i,j \in \sigma.$ $\sgn s_k^\sigma$ for $k \in  [n]\setminus\sigma$ can be anything.
			(B) A subset $\sigma$ is forbidden if and only if there is a mix of signs in it. More precisely, if it exists $i,j \in \sigma$ such that $\sgn s_i^\sigma = -\sgn s_j^\sigma.$ 
			(C) A subset $\sigma$ supports a fixed point of the network if and only if $\sgn s_i^\sigma = \sgn s_j^\sigma = \idx(\sigma) \text{ for all } i,j \in \sigma$ and  $\sgn s_i^\sigma = -\sgn s_k^\sigma  \text{ for all } k \notin \sigma.$}
		\label{fig:sgn-conditions}
	\end{figure}
	
	Two other results from the general TLN theory will be needed. The first relates the indices of all fixed points:
	\begin{theorem}[parity, \cite{curto2019a}]\label{thm:parity}
		Let $(W, \bm{b})$ be a TLN. Then
		$$
		\sum_{\sigma \in \FP(W, \bm{b})} \idx(\sigma)=+1 .
		$$
		In particular, the total number of fixed points $|\FP(W, \bm{b})|$ is always odd.
	\end{theorem}
	
	The second establishes how fixed points behave when nodes are added to the network.
	\begin{corollary}[inheritance, \cite{curto2019a}]\label{cor:inheritance}
		Let $(W,\bm{b})$  be a TLN on $n$ neurons, and let $\sigma\subseteq [n]$. The following are equivalent:
		\begin{enumerate}
			\item $\sigma \in \FP(W,\bm{b})$
			\item $\sigma \in \FP(W_{\tau},\bm{b}_{\tau})$ for all $\sigma \subseteq \tau \subseteq [n]$
			\item $\sigma \in \FP(W_{\sigma},\bm{b}_{\sigma})$ and $\sigma \in \FP(W_{\sigma \cup k},\bm{b}_{\sigma \cup k})$ for all $k \notin \sigma$
			\item $\sigma \in \FP(W_{\sigma \cup k},\bm{b}_{\sigma \cup k})$ for all $k \notin \sigma$
		\end{enumerate}
	\end{corollary}
	
	From the above corollary, a subset $\sigma$ will  support a fixed point if and only if it it is permitted and \emph{survives} the addition of all other nodes in the network (by (a) $\Leftrightarrow$ (c) above). If $\sigma$ is permitted but $\sigma \notin \FP(W_{\sigma \cup k},\bm{b}_{\sigma \cup k})$ for some $k \notin \sigma$, we say that $\sigma$ \emph{dies} in the larger network. 
	
	This concludes the technical background for inhibition-dominated TLNs. We now turn to domination, a theory specific to gCTLNs, which is an important graph-based simplification tool for the dynamics.
	
	\subsubsection{Domination background}
	
	The domination results introduced in this section hold for gCTLNs, and were originally introduced in \cite{curto2025a}. Formally, we say that a TLN $(W,\theta)$ defines a \emph{generalized combinatorial threshold-linear network} (gCTLN), denoted  $(G, \bm{\varepsilon}, \bm{\delta})$, if $W$ is prescribed by a directed graph $G$ as follows:
	\begin{equation}\label{eq:generalized-synapses}
		W_{ij} = \begin{cases}
			0 & \text{if } i = j, \\
			-1 + \varepsilon_j & \text{if } j \rightarrow i \text{ in } G, \\
			-1 - \delta_j & \text{if } j \not\rightarrow i \text{ in } G,
		\end{cases}
	\end{equation}	
	with $\varepsilon_j,\delta_j>0$, $\varepsilon_j<1$ for all $j \in [n]$, and $\theta>0$. We omit $\theta$ in the gCTLN notation because for uniform input, the exact value of $\theta$ only scales the value of the fixed points, and thus supports do not change \cite{morrison2019,curto2025a}. When $\varepsilon_j = \varepsilon, \delta_j = \delta$ for all $i$, we obtain a \textit{combinatorial threshold-linear network} (CTLN), denoted $(G, \varepsilon, \delta)$.
	
	For these graph-based networks, we can use the concept of domination to simplify its set of fixed point supports:
	\begin{definition}[domination,\cite{curto2025a}]\label{def:domination}
		Let $j, k \in[n]$ be vertices of $G$. We say that $k$ graphically dominates $j$ in $G$ if the following two conditions hold:
		\begin{enumerate}
			\item For each vertex $i \in[n] \backslash\{j, k\}$, if $i \rightarrow j$ then $i \rightarrow k$.
			\item $j \rightarrow k$ and $k \nrightarrow j$.
		\end{enumerate}
		If there exists a $k$ that graphically dominates $j$, we say that $j$ is a \emph{dominated node} (or dominated vertex) of $G$. If $G$ has no dominated nodes, we say that it is \emph{domination free}.
	\end{definition}
	
	When a gCTLN is defined by a graph that contains domination relations, dominated nodes can be removed from the network without changing the set of fixed point supports:
	
	\begin{theorem}\label{thm:removal-of-dominated-nodes}
		Let $G$ be a graph with vertex set $[n]$, and suppose $j \in[n]$ is a dominated node in $G$. Then, for any choice of gCTLN parameters $\bm{\varepsilon}, \bm{\delta}$ and $\theta$,		
		$$
		\FP(G,\bm{\varepsilon}, \bm{\delta})=\FP(G|_{[n] \backslash j},\bm{\varepsilon}_{[n] \backslash j}, \bm{\delta}_{[n] \backslash j})
		$$
		In such a case, we say that $j$ is a \emph{removable} node.
	\end{theorem}
	
	All other domination-related results needed here derive from this theorem, and will be introduced and proven as needed in subsequent sections. With this, and the general TLN tools, we are now well-equipped to prove our main results.
	
	We begin by introducing a concept that is the basis for the proofs in this section, it specifies the relationship that a given component has with the rest of the graph:
	\begin{definition}\label{def:simply-embedded-component}
		Let $(W,\theta)$ be a TLN with $n$ nodes. We say that $\tau \subseteq [n]$ is \emph{simply-embedded} in the network if for every $j \notin \tau$, $W_{ij} = W_{kj} = \gamma_j $ for all $i,k \in \tau$. That is, if $\tau = \{1,\dots,\ell\}$ and $[n]\setminus\tau = \{\ell+1,\dots,n\}$, then the matrix $W$ has the form:
		\begin{align}
			W = \left[\begin{array}{c|ccc}
				& \vert &       & \vert \\
				W_{\tau} & \gamma_{\ell+1} & \cdots & \gamma_{n} \\
				& \vert &       & \vert \\
				\hline
				\ast & \multicolumn{3}{c}{W_{[n]\setminus\tau}} 
			\end{array}\right]
		\end{align}
	\end{definition}	

	Note that a low-rank gluing of components $(W_1,\theta), \dots, (W_N,\theta)$ with $W_i = W_{\tau_i,\tau_i}$ defines a partition $\tau_1, \dots, \tau_N$ of the vertex set $[n]$ in which every component is simply-embedded.  This concept was originally introduced in \cite{parmelee2022a} as a \emph{simply-embedded partition}. Similarly,  a rank-1 gluing defines a partition in which not only every component is simply-embedded, but also $[n]\setminus\tau_i$ is simply-embedded. This was termed a \emph{strongly simply-embedded partition} in \cite{parmelee2022a}. Thus, our rank-1 and low-rank gluings are a direct generalization of these concepts from CTLNs to the broader class of TLNs. 
	
	In this section, we prove that these simply-embedded structures allows us to factorize the determinant quantities $s_i^\sigma$  (Eq. \eqref{eq:s_i}), which in turn lets us relate the fixed points of the whole network to the fixed points of its components. This is made possible by the central technical tool for this section, a determinant factorization for matrices with a rank 1 block. This lemma is a common feature in proofs of this kind \cite{curto2019a, parmelee2022a}, but has never been explicitly presented on its own, as we do next. 
	\begingroup
	\def\thelemma{\ref{lemma:det-linear-algebra}} 
	\begin{lemma}
		Let 
\begin{figure}[H]
	\centering
	\includegraphics[width=0.15\textwidth]{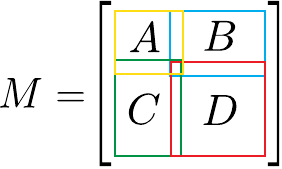}
\end{figure}
\noindent be an $(m+n-1)\times(m+n-1)$ matrix, consisting of the blocks $A$ ($m\times m$), $B$ ($m\times n$), $C$ ($n\times m$) and $D$ ($n\times n$),	where adjacent blocks overlap in one row/column. Note that $a_{m, m}=b_{m,1}=c_{1,m}=d_{1,1}$ is the single entry where all four matrices overlap. If $B$ or $C$ is rank 1 and $a_{m,m} \neq 0$, then $$\det M=\frac{1}{a_{m,m}} \det A \det D.$$
	\end{lemma}
	\addtocounter{lemma}{-1}
	\endgroup

	\begin{proof}
		Let \begin{align*}
			A&=\left[\begin{array}{ccc}
				a_{1,1} & \cdots & a_{1, m} \\
				\vdots & & \vdots \\
				a_{m, 1} & \cdots & a_{m, m}
			\end{array}\right], 
			B=\left[\begin{array}{ccc}
				b_{1,1} & \cdots & b_{1, n} \\
				\vdots & & \vdots \\
				b_{m, 1} & \cdots & b_{m, n}
			\end{array}\right], \\
			C&=\left[\begin{array}{ccc}
				c_{1,1} & \cdots & c_{1, m} \\
				\vdots & \cdots & \vdots \\
				c_{n, 1} & \cdots & c_{n, m}
			\end{array}\right], 
			D=\left[\begin{array}{ccc}
				d_{1,1} & \cdots & d_{m, n} \\
				\vdots & \cdots & \vdots \\
				d_{n, 1} & \cdots & d_{n, n}
			\end{array}\right].
		\end{align*} 
		Note that  
		\begin{align*}
			\left[\begin{array}{c}
				a_{1, m} \\
				\vdots \\
				a_{m, m}
			\end{array}\right]=\left[\begin{array}{c}
				b_{1,1} \\
				\vdots \\
				b_{m, 1}
			\end{array}\right], 
			\left[\begin{array}{c}
				c_{1, m} \\
				\vdots \\
				c_{n, m}
			\end{array}\right]=\left[\begin{array}{c}
				d_{1,1} \\
				\vdots \\
				d_{n, 1}
			\end{array}\right],\left[\begin{array}{llll}
				b_{m, 1} & \cdots & b_{m, n}
			\end{array}\right]=\left[\begin{array}{llll}
				d_{1,1} & \cdots & d_{1, n}
			\end{array}\right]
		\end{align*}  and $\left[\begin{array}{llll}
			c_{1,1} & \cdots & c_{1, m}
		\end{array}\right] = \left[\begin{array}{llll}
			a_{m,1} & \cdots & a_{m,m}
		\end{array}\right]$. 
		Assume first that $B$ is rank 1, then we can take all columns of $B$ to be multiples of its first column. Let $\beta_k$ be such that 
		$$
		\left[\begin{array}{c}
			b_{1, k} \\
			\vdots \\
			b_{m, k}
		\end{array}\right]=\beta_k\left[\begin{array}{c}
			b_{1,1} \\
			\vdots \\
			b_{m, 1}
		\end{array}\right],
		$$ with $\beta_1 = 1$. Subtracting $\beta_k$ times the $m$-th column of $M$ from its $k$-th column, we get
		\begin{align}
			\det M &= \left[\begin{array}{ccc|ccc}
				a_{1,1} & \cdots & a_{1, m} & 0 & \cdots & 0 \\
				\vdots & & \vdots & \vdots & & \vdots \\
				a_{m, 1} & \cdots & a_{m, m} & 0 & \cdots & 0 \\
				\hline
				c_{1,1} & \cdots & d_{2,1} & d_{2,2}-\beta_2 d_{2,1} & \cdots & d_{2, n}-\beta_n d_{2,1} \\
				\vdots & & \vdots & \vdots & & \vdots \\
				c_{n, 1} & \cdots & d_{n, 1} & d_{n, 2}-\beta_2 d_{n, 1} & \cdots & d_{n, n}-\beta_n d_{n, 1}
			\end{array}\right] \notag \\
			\displaybreak[2]
			&= \det A \det\left[\begin{array}{cccc}
				d_{2,2}-\beta_2 d_{2,1} & \cdots & d_{2, n}-\beta_n d_{2,1} \\
				\vdots & & \vdots \\
				d_{n, 2}-\beta_2 d_{n, 1} & \cdots & d_{n, n}-\beta_n d_{n, 1}
			\end{array}\right] \notag \\
			\displaybreak[2]
			&= \det A \det\left[\begin{array}{c|ccc}
				1 & 0 & \cdots & 0 \\
				\hline
				& d_{2,2}-\beta_2 d_{2,1} & \cdots & d_{2, n}-\beta_n d_{2,1} \\
				*  & \vdots & & \vdots \\
				& d_{n, 2}-\beta_2 d_{n, 1} & \cdots & d_{n, n}-\beta_n d_{n, 1}
			\end{array}\right] \notag \\
			\displaybreak[2]
			&= \det A \frac{1}{b_{m, 1}} \det\left[\begin{array}{c|ccc}
				b_{m, 1} & 0 & \cdots & 0 \\
				\hline
				d_{2,1} & d_{2,2}-\beta_2 d_{2,1} & \cdots & d_{2, n}-\beta_n d_{2,1} \\
				\vdots  & \vdots & & \vdots \\
				d_{n, 1} & d_{n, 2}-\beta_2 d_{n,1} & \cdots & d_{n, n}-\beta_n d_{n, 1}
			\end{array}\right] \notag \\
			\displaybreak[2]
			&= \frac{1}{b_{m, 1}} \det A \det\left[\begin{array}{cccc}
				b_{m, 1} & \beta_2 b_{m, 1} & \cdots & \beta_n b_{m, 1} \\
				d_{2,1} & d_{2,2} & \cdots & d_{2, n} \\
				\vdots  & \vdots & & \vdots \\
				d_{n, 1} & d_{n, 2} & \cdots & d_{n, n}
			\end{array}\right] \notag \\
			\displaybreak[2]
			&= \frac{1}{a_{m, m}} \det A \det\left[\begin{array}{cccc}
				d_{1,1} & d_{1,2} & \cdots & d_{1, n} \\
				d_{2,1} & d_{2,2} & \cdots & d_{2, n} \\
				\vdots & \vdots & & \vdots \\
				d_{n, 1} & d_{n, 2} & \cdots & d_{n, n}
			\end{array}\right] = \frac{1}{a_{m, m}} \det A \det D \notag
		\end{align}
	\end{proof}

	\subsection{Proof of Theorem \ref{thm:low-rank-gluing-menu}: low-rank gluings}\label{sec:low-rank-gluing}
	
	The $s_i^\sigma$ quantities for a simply-embedded component $\tau$ will have precisely the overlapping, rank-1 block structure required by Lemma \ref{lemma:det-linear-algebra}. This allows us to factor the $s_i^\sigma$ quantities into the simply-embedded part, $\sigma\cap\tau$, and the rest of the graph $\sigma\cap\omega$, for $\omega = [n]\setminus\tau$. This factorization bridges the simply-embedded structure of the connectivity matrix with statements on the network dynamics via Theorem \ref{thm:sgn-conditions}.
	\begin{theorem}\label{thm:simply-embedded-si-factorization}
		Let $\tau$ be simply-embedded in $(W,\theta)$. Let $\omega = [n]\setminus\tau$. Then for any $\sigma \subseteq [n]$, we have
		$$s_i^\sigma = \frac{1}{\theta}s_i^{\sigma\cap\omega} s^{\sigma\cap\tau}_i = \alpha s^{\sigma\cap\tau}_i \quad\text{for each $i \in \tau$,}$$
		where $\alpha = \frac{1}{\theta}s_i^{\sigma\cap\omega}$ has the same value for every $i \in \tau$.
	\end{theorem}
	\begin{proof}
		Since $\tau$ is simply-embedded with uniform input, $W$ has the form
		$$W=\left[\begin{array}{c|c}
			W_\tau & \mathbbm{1}_\tau \bm{\gamma}^\intercal \\
			\hline
			* & W_\omega
		\end{array}\right],$$ where $\mathbbm{1}_\tau$ is a column vector of all ones indexed by $\tau$, and $\bm{\gamma}^\intercal = \left[\begin{array}{lll}
			\gamma_{1} & \cdots & \gamma_{m} 
		\end{array}\right]$ so that
		$$\mathbbm{1}_\tau \bm{\gamma}^\intercal=\left[\begin{array}{lll}
			\vert &  & \vert \\
			\gamma_1 & \cdots & \gamma_m \\
			\vert &  &  \vert
		\end{array}\right].$$
		Let $i \in \tau$, then
		\begin{align*}
			s_{i}^{\sigma} & = \det\left(\left(I-W_{\sigma\cup\{i\}}\right)_{i};\theta\mathbbm{1}\right) 
			=\det\left[\begin{array}{c|c|ccc}
				& \vert & \gamma_1 & & \gamma_m \\
				I-W_{\sigma\cap\tau} & \theta & \vdots & \cdots & \vdots \\
				& \vert & \gamma_1 &  & \gamma_m \\
				\cline{1-2}
				\rotvert \, (I-W_{\sigma\cap\tau})_i \, \rotvert & \theta & \gamma_1 & \cdots & \gamma_m\\
				\hline
				& \vert & & & \\
				* & \theta& & I-W_{\sigma\cap\omega} & \\
				& \vert & & & \\
			\end{array}\right]
		\end{align*}
		Since the upper right block,  $\mathbbm{1}_{\sigma \cap \tau} \left[\begin{array}{cccc}
			\theta & \gamma_{1} & \cdots & \gamma_m
		\end{array}\right]$, is rank 1 and $\theta \neq 0$, by Lemma \ref{lemma:det-linear-algebra}, 
		\begin{align*}
			s_{i}^{\sigma} & = \frac{1}{\theta} \det \left[\begin{array}{c|c}
				\multirow{ 2}{*}{$I-W_{(\sigma \cap \tau)\setminus\{i\}}$}  & \vert \\
				& \theta \\
				\cline{1-1}\rotvert (I-W_{\sigma\cap\tau})_i \rotvert & \vert \\
			\end{array}\right] 
			\det\left[\begin{array}{c|ccc}
				\theta & \gamma_1 & \cdots & \gamma_n\\
				\hline
				\vert & & & \\
				\theta & & I-W_{\sigma\cap\omega}& \\
				\vert & & & \\
			\end{array}\right] \\
			& =\frac{1}{\theta}  \det((I-W_{\sigma\cap\tau})_i;\theta\mathbbm{1}) \det((I-W_{(\sigma\cap\omega) \cup \{i\}})_i;\theta\mathbbm{1})	\\
			& = \frac{1}{\theta}s_i^{\sigma\cap\tau} s^{\sigma\cap\omega}_i
		\end{align*}
		note that the last matrix does not depend on $i$.	Making $\alpha = \frac{s_i^{\sigma\cap\omega}}{\theta}$, the result follows.
	\end{proof}
	
	We have thus showed that the $s_i^\sigma$ values factorize for simply-embedded components. When each component is simply-embedded, the signs of the $s_i$'s are inherited to and from each component:
	\begin{lemma}\label{lemma:simply-embedded-inheritance}
		Let $(W,\theta)$ be a low-rank gluing of $(W_1,\theta), \dots, (W_N,\theta)$, with $W_{\tau_i,\tau_i} = W_i$. For any $\sigma \subseteq [n]$, let $\sigma_\ell \od \sigma \cap \tau_\ell$.  Then for any $\sigma_i \neq \emptyset$,
		$$\sgn s_i^\sigma = \sgn s^\sigma_j \quad \Leftrightarrow \quad \sgn s_i^{\sigma_\ell} = \sgn s_j^{\sigma_\ell}, \quad\text{for all $i,j \in \tau_\ell$}.$$
	\end{lemma}
	\begin{proof}
		Since each $\tau_\ell$ is simply-embedded in $(W,\theta)$, and input is assumed to be uniform across all components, by Theorem \ref{thm:simply-embedded-si-factorization} we have $s_i^\sigma = \alpha s_i^{\sigma_\ell}$ for all $i \in \tau_\ell$, where $\alpha = \frac{1}{\theta}s_i^{\sigma\setminus\tau_\ell}$ has the same value for every $i \in \tau_\ell$.
		Hence, for all $i,j \in \tau_\ell$, we have that $\sgn s_i^\sigma = \sgn s_j^\sigma$ if and only if  $\sgn \alpha s_i^{\sigma_\ell} = \sgn \alpha s_j^{\sigma_\ell}$  if and only if $\sgn s_i^{\sigma_\ell} = \sgn s_j^{\sigma_\ell}$.	
	\end{proof}
	
	This sign-inheritance is the final piece needed to prove Theorem \ref{thm:low-rank-gluing-menu}, as it allows for a straightforward application of Theorem \ref{thm:sgn-conditions} to check for fixed point survival.

	\begingroup
	\def\thetheorem{\ref{thm:low-rank-gluing-menu}} 
	\begin{theorem}[fixed point compositionality]
		Let $(W,\theta)$ be a low-rank gluing of TLNs $(W_1,\theta), \dots, (W_N,\theta)$. Then, for each $\sigma \in \FP(W,\theta)$, $\sigma$ has the form: 
		$$\sigma = \bigcup_{\ell=1}^{N} \sigma_\ell, \text{ where } \sigma_\ell \in \FP(W_\ell,\theta) \cup \{\emptyset\}.$$ That is, the fixed point supports of $(W,\theta)$ are unions of the fixed point supports of component networks, at most one per component.
	\end{theorem}
	\addtocounter{theorem}{-1}
	\endgroup
	
	\begin{proof}
		Let  $W_{\tau_i,\tau_i} = W_i$ and $\sigma_\ell \od \sigma \cap \tau_\ell$. Then, for $\sigma \in \FP(W,\theta)$, we have
		$$\sgn s_i^\sigma = \sgn s_j^\sigma = -\sgn s_k^\sigma$$
		for any $i,j \in \sigma_\ell$ and $k \in \tau_\ell\setminus\sigma_\ell$, by Theorem~\ref{thm:sgn-conditions} (sign conditions). Then by Lemma \ref{lemma:simply-embedded-inheritance}, we see that whenever $\sigma_\ell \neq \emptyset$,
		$$\sgn s_i^{\sigma_\ell} = \sgn s_j^{\sigma_\ell} = -\sgn_k^{\sigma_\ell},$$  and so $\sigma_\ell$ satisfies the sign conditions in $(W_{\tau_\ell},\theta)$. Thus $\sigma_\ell \in \FP(W_{\tau_\ell},\theta)$ for every nonempty $\sigma_\ell$.
	\end{proof}

	While Theorem \ref{thm:low-rank-gluing-menu} shows that a global fixed point must be a union of component fixed points, it does not say which unions are indeed global fixed points. Next, we prove a `weak reverse' of Theorem \ref{thm:low-rank-gluing-menu}, showing that if a given support is a union of component fixed points, and in addition it is permitted, then this fixed point it will be a global fixed point of the network. 
	\begin{lemma}\label{lemma:weak-reverse-menu}
		Let $(W,\theta)$ be a low-rank gluing of $(W_1,\theta), \dots, (W_N,\theta)$, with $W_{\tau_i,\tau_i} = W_i$, and let $\sigma_i = \sigma \cap \tau_i$. If $\sigma \in \FP(W_\sigma, \theta)$ and $\sigma_i \in \FP(W_{\tau_i}, \theta)$, then $ \sigma \in \FP(W, \theta)$.
	\end{lemma}
	\begin{proof}
		Note $\sigma_i \in \FP(W_{\tau_i}, \theta)$ implies $\sigma_i \neq \emptyset$ for all $i \in [N]$, and
		$$\sgn s_j^{\sigma_i} = \sgn s_k^{\sigma_i} = -\sgn s_l^{\sigma_i} \quad \forall j,k \in \sigma_i, \ l \in \tau_i \setminus \sigma_i$$
		by sign conditions.
		
		On the other hand, $\sigma \in \FP(W_\sigma, \theta)$ implies $\sgn s_j^{\sigma} = \sgn s_k^{\sigma}$ for all $j,k \in \sigma$. By Lemma \ref{lemma:simply-embedded-inheritance}, we get:
		$$\sgn s_l^{\sigma} = -\sgn s_j^{\sigma} \quad \forall l \in [n] \setminus \sigma, \ j \in \sigma.$$
		
		By Theorem~\ref{thm:sgn-conditions}, it follows that $\sigma \in \FP(W, \theta)$.
	\end{proof}
	
	Putting these two together, we get a weak characterization of the fixed points of a low-rank gluing:
	\begin{lemma}\label{lemma:weak-fixed-point-partition}
		Let $(W,\theta)$ be a low-rank gluing of $(W_1,\theta), \dots, (W_N,\theta)$, with $W_{\tau_i,\tau_i} = W_i$, and let $\sigma_i = \sigma \cap \tau_i$. Given $\sigma \subseteq [n]$ such that $\sigma_i \neq \emptyset$ for all $i \in [N]$, we have:
		$$
		\sigma \in \FP(W, \theta) \quad \Leftrightarrow \quad \sigma \in \FP(W_\sigma, \theta) \text{ and } \sigma_i \in \FP(W_{\tau_i}, \theta) \quad \forall i \in [N].
		$$
	\end{lemma}
	\begin{proof}
		Follows from Theorem \ref{thm:low-rank-gluing-menu} and Lemma \ref{lemma:weak-reverse-menu}.
	\end{proof}
	
	In the next section, we get stronger results for the characterization of the fixed points of rank-1 gluings.
	
	\subsection{Proof of Theorem \ref{thm:rank-1-gluings}: rank-1 gluings}\label{sec:rank-1-gluing}
	
	Recall from Definitions \ref{def:low-rank-gluings} and \ref{def:simply-embedded-component} that a rank-1 gluing of the component networks $(W_{\tau_1},\theta), \dots, (W_{\tau_N},\theta)$, with $W_{\tau_i,\tau_i} = W_i$, defines a partition $\{\tau_1,\cdots,\tau_N\}$ of the vertex set $[n]$ in which, in addition to each $\tau_\ell$ being simply-embedded in $(W,\theta)$, its complement $[n]\setminus\tau_\ell$ is also simply-embedded. 
	
	This stricter structure now allows for a full factorization of the $s_i^\sigma$ quantities. This complete factorization will allow us to prove Theorem \ref{thm:rank-1-gluings}, which explicitly characterizes the fixed points of the glued network in terms of component fixed points that either survive or die in the full network. To that end, we introduce notation for the sets of surviving and dying fixed point supports for  any subnetwork $(W_\sigma,\theta)$ of a given TLN $(W,\theta)$, as
	$$S_\sigma \od \FP(W_\sigma,\theta) \cap \FP(W,\theta), \quad \text{and} \quad D_\sigma \od \FP(W_\sigma,\theta)\setminus S_\sigma,$$ respectively.
	
	We begin, as we did in the last section, by proving a factorization lemma for the $s_i^{\sigma}$ quantities. As before, this will be the main tool for proving the full characterization via Theorem \ref{thm:sgn-conditions}.
	
	\begin{lemma}\label{lemma:full-factorization}
		Let $(W,\theta)$ be a rank-1 gluing of $(W_1,\theta), \dots, (W_N,\theta)$, with $W_{\tau_i,\tau_i} = W_i$. For any $\sigma \subseteq [n]$, denote $\sigma_\ell \od \sigma \cap \tau_\ell$, and $\sigma_{\ell_1\dots \ell_k} \od \sigma_{\ell_1} \cup \dots \cup \sigma_{\ell_k}$ and let $I = \{\ell \in [N]~|~ \sigma_\ell \neq \emptyset\}$. Then for every $i\in [n]$, we have 
		$$s_i^\sigma = \frac{1}{\theta^{|I|-1}}\prod_{\ell\in I} s_i^{\sigma_\ell},$$
		where $s_i^{\sigma_\ell}$ has the same value for every $i \in [n]\setminus \tau_\ell$. Moreover, for any $\sigma_\ell \in \FP(W_{\tau_\ell},  \theta)$ and $j \in \tau_\ell$:
		$$\sgn s_j^{\sigma_\ell} = \begin{cases}
			\phantom{-}\idx(\sigma_\ell) & \text{if } j \in \sigma_\ell \\
			-\idx(\sigma_\ell) & \text{if } j \in \tau_\ell\setminus \sigma_\ell
		\end{cases}$$
		while for any $k\notin \tau_\ell$,
		$$\sgn s_k^{\sigma_\ell} = \begin{cases}
			-\idx(\sigma_\ell) & \text{if } \sigma_\ell \in S_{\tau_\ell}\\
			\phantom{-}\idx(\sigma_\ell) & \text{if } \sigma_\ell \in D_{\tau_\ell}.
		\end{cases}$$	
	\end{lemma}
	
	\begin{proof}
		Since $\tau_1$ is simply-embedded in $[n]\setminus\tau_1$, $$s_i^\sigma = \frac{1}{\theta} s_i^{\sigma_{2\dots N}}s_i^{\sigma_1} \text{ for all } i \in \tau_1$$
		by Theorem~\ref{thm:simply-embedded-si-factorization}.  On the other hand, since $[n]\setminus\tau_1$ is also simply-embedded, we also have
		$$s_i^\sigma = \frac{1}{\theta} s_i^{\sigma_1} s_i^{\sigma_{2\dots N}} \text{ for all } i \in [n]\setminus\tau_1.$$
		Therefore, the above factorization holds for all $i \in [n]$. Similarly, since both $\tau_2$ and $[n]\setminus\tau_2$ are simply-embedded, 
		$$s_i^{\sigma_{2\dots N}}= \frac{1}{\theta} s_i^{\sigma_2} s_i^{\sigma_{3\dots N}} \text{ for all } i \in [n]$$
		by Theorem~\ref{thm:simply-embedded-si-factorization}, and so $s_i^\sigma = \frac{1}{\theta^2} s_i^{\sigma_1} s_i^{\sigma_2} s_i^{\sigma_{3\dots N}}$.  Continuing in this fashion, we see that for any $i \in [n]$,
		$$s_i^{\sigma} =\frac{1}{\theta^{N-1}} s_i^{\sigma_1}\dots s_i^{\sigma_N}.$$
		Note that if $\sigma_\ell = \emptyset$, then $s_i^{\sigma_\ell} = s_i^{\emptyset} = s_i^{\{j\}} = \theta$, and thus for all $i \in [n]$,
		$$s_i^{\sigma} = \frac{\theta^{N-|I|}}{\theta^{N-1}}\prod_{i \in I}s_i^{\sigma_\ell} = \frac{1}{\theta^{|I|-1}}\prod_{i\in I} s_i^{\sigma_\ell}.$$
		The fact that $s_i^{\sigma_\ell}$ has the same value for every $i \in [n]\setminus \tau_\ell$ is a direct consequence of Theorem \ref{thm:simply-embedded-si-factorization}.
		
		Finally, to prove the last statements about the signs of $s_i^{\sigma_\ell}$, observe that for $i \in \tau_\ell$, the values of $\sgn s_i^{\sigma_\ell}$ are fully determined by Theorem \ref{thm:sgn-conditions} (sign conditions) since $\sigma_\ell \in \FP(G|_{\tau_\ell})$ by hypothesis.   In particular, if $\sigma_\ell \in S_{\tau_\ell}$, then $\sigma_\ell$ survives the addition of every $k \notin \tau_\ell$, and so $\sgn s_k^{\sigma_\ell} = -\idx(\sigma_\ell)$ by the same theorem.  On the other hand, if $\sigma_\ell \in D_{\tau_\ell}$ then $\sigma_\ell$ dies in $G$ and so there is some $k \notin \tau_\ell$ for which $\sgn s_k^{\sigma_\ell} = \idx(\sigma_\ell)$.  But by the first part of the theorem, all the $s_k^{\sigma_\ell}$ values are identical for $k \in [n] \setminus \tau_\ell$, and thus $\sgn s_k^{\sigma_\ell} = \idx(\sigma_\ell)$ for all such $k$.
	\end{proof}
	
	We are now ready to prove Theorem \ref{thm:rank-1-gluings}, which generalizes results originally presented in \cite{curto2019a,parmelee2022a}. The proof is also a generalized version of the arguments used in those earlier works.
	
	\begingroup
	\def\thetheorem{\ref{thm:rank-1-gluings}} 
	\begin{theorem}
		Let $(W,\theta)$ be a rank-1 gluing of $(W_1,\theta), \dots, (W_N,\theta)$.  Then $\sigma \in \FP(W,\theta)$ if and only if $\sigma = \bigcup_{\ell=1}^{N} \sigma_\ell$ for $\sigma_\ell \in \FP(W_\ell,\theta) \cup \{\emptyset\}$, and either
		\begin{enumerate}
			\item none of the $\sigma_\ell$ are in $\FP(W,\theta) \cup \{ \emptyset\}$, or
			\item every $\sigma_\ell$ is in $\FP(W,\theta) \cup \{ \emptyset\}$.
		\end{enumerate}
		In other words, $\sigma \in \FP(W,\theta)$ if and only if $\sigma$ is either a union of dying fixed points, exactly one from every component or it is a union of surviving fixed points $\sigma_i$, at most one per component.
	\end{theorem}
	\addtocounter{theorem}{-1}
	\endgroup
	
	\begin{proof}
		Let $W_{\tau_i,\tau_i} = W_i$ and $\sigma_\ell \od \sigma \cap \tau_\ell$. Assuming without loss of generality that $\theta=1$, by Lemma \ref{lemma:full-factorization} we have that for $I \od \{i ~|~ \sigma_i \neq \emptyset \}$ and $\sigma\subseteq[n]$,
		\begin{equation}\label{eq:full_factorization}
			s_j^\sigma = \prod_{i \in I} s_j^{\sigma_i},
		\end{equation}
		with $s_j^{\sigma_i}$ being constant for all $j \in [n]\setminus\tau_i$, for each $i \in [N]$ (Fig. \ref{fig:first_si_factorization_cartoon}).
		
		($\Rightarrow$) Suppose $\sigma \in \FP(W,\theta)$, so that 
		$$
		\sgn s_i^\sigma =
		\begin{cases}
			\idx \sigma & \text{if } i \in \sigma, \\
			-\idx \sigma & \text{if } i \notin \sigma. \\
		\end{cases}
		$$
		By Lemma \ref{lemma:full-factorization}, $\sgn s_i^{\sigma} = \prod_{i \in I} \sgn s_j^{\sigma_i}$ (Fig. \ref{fig:first_si_factorization_cartoon}). But by Theorem \ref{thm:low-rank-gluing-menu}, $\sigma_i \in \FP(W_{\tau_i},\theta)$ for every $i \in I$; and by the $\sgn$ formula of Lemma \ref{lemma:full-factorization},  if we denote by $\S \od \{a\in I~|~\sigma_a \in S_a\}$, and fix $j \in \sigma$, we get that for $j \in \sigma_i$: 
		\begin{equation}\label{eq:sgn-s_j}
			\sgn s_j^\sigma = \idx(\sigma_i) \prod_{a \in \S\setminus\{i\}} -\idx(\sigma_a)  \prod_{b \in \S^c\setminus\{i\}} \idx(\sigma_b) = (-1)^{|\S \setminus \{i\}|}\prod_{\ell\in I} \idx(\sigma_\ell),
		\end{equation}
		\begin{figure}[H]
			\centering
			\includegraphics[width=\textwidth]{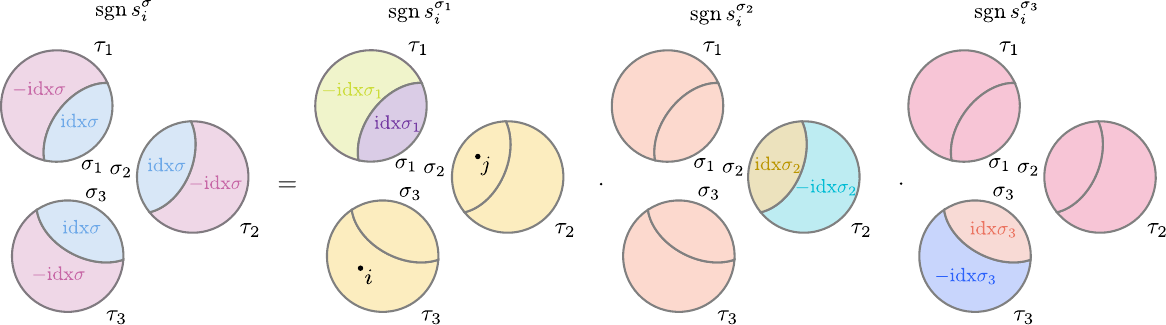}
			\caption{Cartoon showing $\sgn s_i^{\sigma} = \prod_{i \in I} \sgn s_j^{\sigma_i}$ for three components. Note that $\sgn s_j^{\sigma_i}$ is constant for any $j \in [n]\setminus\tau_i$, for all $i \in [N]$. 
				On the left hand side, we have the index equalities because $\sigma \in \FP(W,\theta)$. On the right hand side, we have the index equalities because $\sigma_i \in \FP(W_{\tau_i},\theta)$.}
			\label{fig:first_si_factorization_cartoon}
		\end{figure}
		
		Now, observe that if $\sigma$ contained a mix of $\sigma_a \in S_a$ and $\sigma_b \in D_b$, then there would be $i, j \in \sigma$ such that $i \in \sigma_a$ for some $a\in \S$, while $j \in \sigma_b$ for some $b \notin \S$ and thus $\sgn s_i^{\sigma} = (-1)^{|S|-1} \prod_{\ell \in I} \idx \sigma_\ell$ and  $\sgn s_j^{\sigma} = (-1)^{|S|} \prod_{\ell \in I} \idx \sigma_\ell = -\sgn s_i^{\sigma}$. But this contradicts the fact that $\sigma \in \FP(W,\theta)$ by Theorem \ref{thm:sgn-conditions} (since we assumed $i,j \in \sigma$). Thus, we must have either $\sigma_i \in S_{\tau_i}$ for all $i \in I$, as in (a), or $\sigma_i \in D_{\tau_i}$ for all $i \in I$ as in (b).
		
		Now to see that when $\sigma_i \in D_{\tau_i}$ for all $i \in I$, we must have $I = [N]$, suppose to the contrary that $I \subsetneq [N]$ so that there is some $m \in [N]$ such that $\tau_m \cap \sigma = \emptyset$ (Fig. \ref{fig:second_si_factorization_cartoon}). Then, for $k\in\tau_m$ (so $k \notin \sigma$), we have $\sgn s_k^{\sigma_\ell} = \idx(\sigma_\ell)$ for all $\ell \in I$, by Lemma~\ref{lemma:full-factorization}, since $\sigma_\ell \in D_{\tau_\ell}$.  Thus
		$\sgn s_k^\sigma = \prod_{\ell \in I} \idx(\sigma_\ell)$.
		\begin{figure}[H]
			\centering
			\includegraphics[width=\textwidth]{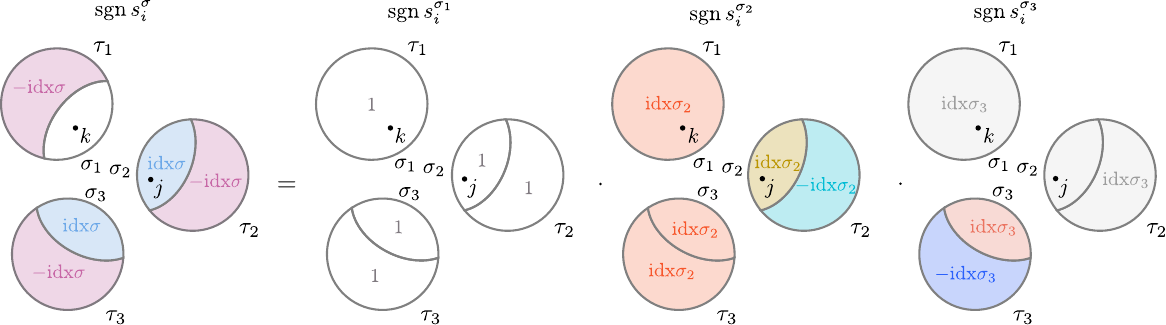}
			\caption{Cartoon showing the assumption that $\sigma_1 = \emptyset$ ($m=1$). This implies $\sgn s_i^{\sigma_1} = 1$. On the right hand side, we have the index equalities because, as in Figure \ref{fig:first_si_factorization_cartoon}, $\sigma_i \in \FP(W_{\tau_i},\theta)$ but now additionally $\sgn s_k^{\sigma_\ell}=\idx \sigma_\ell$.}
			\label{fig:second_si_factorization_cartoon}
		\end{figure}
		
		On the other hand, if $j \in \sigma$, with $j \in \tau_i$, we have $$\sgn s_j^{\sigma} =(-1)^{|\S \setminus \{i\}|}\prod_{\ell\in I} \idx(\sigma_\ell) = \prod_{\ell\in I} \idx(\sigma_\ell).$$ Where the first equality follows from Equation \eqref{eq:sgn-s_j}, and the second because $\S = \emptyset$ in this case.
		
		But then $\sgn_k^\sigma = \sgn_j^\sigma$ for $k \notin \sigma$ and $j \in \sigma$, which contradicts $\sigma\in\FP(W,\theta)$ by Theorem \ref{thm:sgn-conditions}. Thus $I=[N]$ when $\sigma_\ell \in D_\ell$ for all $\ell \in I$.
		
		($\Leftarrow$) Suppose (a) holds and so $\sigma_i \in S_{\tau_i}$ for all $i \in I$. Let us check the sign conditions for $\sigma \od \bigcup_{i \in I} \sigma_i$. 
		
		For any $j \in \sigma$, there exists $i \in I$ such that $j \in \tau_i$ (Fig. \ref{fig:third_si_factorization_cartoon}).  Then by Equation \eqref{eq:sgn-s_j}, we have
		$$\sgn s_j^{\sigma} = (-1)^{|\S \setminus \{i\}|}\prod_{\ell\in I} \idx(\sigma_\ell) = (-1)^{|I|-1}\prod_{\ell\in I}\idx\sigma_\ell,$$
		since $\S = I$ in this case.
		\begin{figure}[H]
			\centering
			\includegraphics[width=\textwidth]{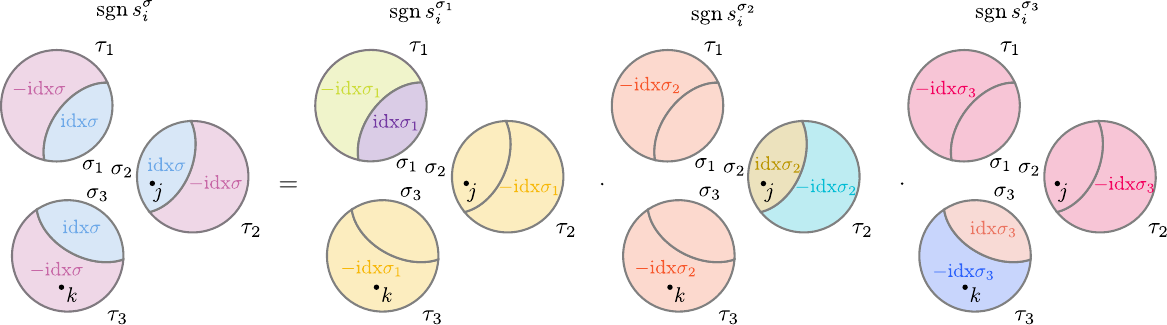}
			\caption{Checking sign conditions for $\sigma \od \bigcup_{i \in I} \sigma_i$, assuming $j \in \tau_2$. On the right hand side, we have the index equalities because $\sigma_i \in \FP(W_{\tau_i},\theta) \cup \{\emptyset\}$ for each $i \in [N]$ and  $\sgn s_k^{\sigma_\ell} = -\idx \sigma_\ell$ for all $\ell \in I$.}
			\label{fig:third_si_factorization_cartoon}
		\end{figure}
		On the other hand, for $k \notin \sigma$, we have $\sgn s_k^{\sigma_\ell} = -\idx \sigma_\ell$ for all $\ell \in I$, by Lemma~\ref{lemma:full-factorization}, since $\sigma_\ell \in S_{\tau_\ell}$.  Thus
		$$\sgn s_k^{\sigma} = \prod_{\ell\in I}(-\idx\sigma_\ell) = (-1)^{|I|}\prod_{\ell\in I}\idx\sigma_\ell = -\sgn s_j^\sigma.$$
		Therefore $\sigma \in \FP(W,\theta)$ by Theorem \ref{thm:sgn-conditions} (sign conditions).
		
		Next, suppose (b) holds so $\sigma_\ell \in D_{\tau_\ell}$ for all $\ell \in [N]$ (so $I = [N]$). Then for any $j \in \sigma$, there is $i \in [N]$ such that $j \in \sigma_i$ and by Equation \eqref{eq:sgn-s_j}, we have
		$$\sgn s_j^{\sigma}=(-1)^{|\S \setminus \{i\}|}\prod_{\ell\in [N]} \idx(\sigma_\ell) = \prod_{\ell\in [N]} \idx(\sigma_\ell),$$
		since $\S=\emptyset$.
		
		Now, let $k \notin \sigma$, with $k \in \tau_m$. Since $\sigma_m \in \FP(W_{\tau_m},\theta)$ with $\sigma_m \neq \emptyset$ (because $I = [N]$), we have $\sgn s_k^{\sigma_m} = -\idx(\sigma_m)$ and thus
		$$\sgn s_k^{\sigma}=\sgn s_k^{\sigma_m} \hspace{-.15in}\prod_{\ell\in [N]\setminus\{m\}}  \hspace{-.15in}\sgn s_k^{\sigma_\ell} = -\idx(\sigma_m)  \hspace{-.15in}\prod_{\ell\in [N]\setminus\{m\}}  \hspace{-.15in}\idx(\sigma_\ell) = -\prod_{\ell\in [N]} \idx(\sigma_\ell) = -\sgn s_j^\sigma.$$
		Thus sign conditions are satisfied, and so $\sigma \in \FP(W,\theta)$.
	\end{proof}
	
	\subsection{Proof of Theorem \ref{thm:generalized-cyclic-union}: cyclic unions}\label{sec:cyclic-unions}
	
		Recall that a for gCTLNs its connectivity matrix is prescribed by a directed graph $G$ and neuron-specific parameters $\varepsilon_j>0$ and $\delta_j>0$:	
	\begin{equation*}
		W_{ij} = \begin{cases}
			0 & \text{if } i = j, \\
			-1 + \varepsilon_j & \text{if } j \rightarrow i \text{ in } G, \\
			-1 - \delta_j & \text{if } j \not\rightarrow i \text{ in } G.
		\end{cases}
	\end{equation*}
	
	Definition \ref{def:simply-embedded-component} easily translates to graph-based networks: a given component $\tau$ is simply-embedded if any node outside of $\tau$ projects identically to every node within $\tau$ (Fig.~\ref{fig:graph-based-se-vs-strongly-se}A). Note that different nodes outside of $\tau$ can have their own constant projection into $\tau$. We can similarly translate low-rank and rank-1 gluings. 	  Low-rank gluings constrain all incoming weights to a component from any other single node to be identical (Fig. \ref{fig:graph-based-se-vs-strongly-se}B), while rank-1 gluings require a much stricter, symmetric condition: all nodes in a component $\tau_i$ must project out identically to all nodes in any other component $\tau_j$ (Fig. \ref{fig:graph-based-se-vs-strongly-se}C).	
	\begin{figure}[!h]
		\begin{center}
			\includegraphics[width=0.75\textwidth]{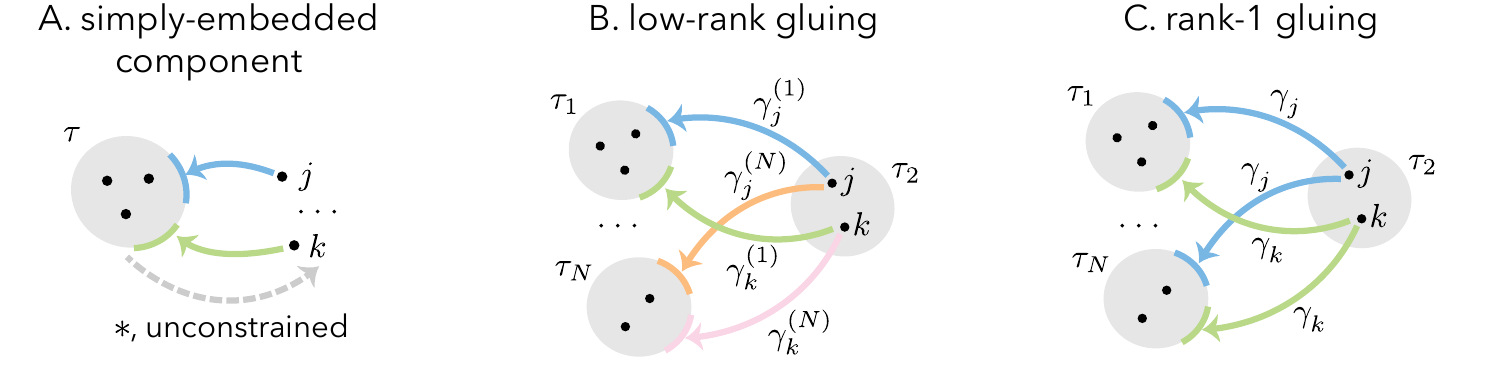}
		\end{center}
		\caption{\textbf{Low-rank gluings versus rank-1 gluings in graph-based networks.} 
			(A) A simply-embedded component $\tau$.
			(B) A low-rank gluing: every node outside of $\tau_i$ treats all of the nodes in $\tau_i$ identically, for all $i \in [N]$.  
			(C) A rank-1 gluing: every node in $\tau_k$ treats all other components $\tau_i$ identically.
		}
		\label{fig:graph-based-se-vs-strongly-se}
	\end{figure}
	
	Recall the cyclic union construction from Figure \ref{fig:CTLN_results_summary}B. Formally, as defined in \cite{curto2019a}:
	
	\begin{definition}
		Let $\{\tau_1, \ldots, \tau_N\}$ be a partition of the index set of a directed graph $G$. We say that $G$ is a \textit{cyclic union} of the induced subgraphs $G|_{\tau_1}, \ldots, G|_{\tau_N}$ if it contains all edges from $\tau_i$ to $\tau_{i+1}$ for $i = 1, \ldots, N-1$, and all edges from $\tau_N$ to $\tau_1$, with no other edges between distinct components $\tau_i, \tau_j$.
	\end{definition}
	
	It can be easily seen that a gCTLN defined by a graph that is a cyclic union will constitute a low-rank gluing of its components (it is not a rank-1 gluing, as a given node will not project identically to all components). Thus, to prove Theorem \ref{thm:generalized-cyclic-union}, we can use Theorem \ref{thm:low-rank-gluing-menu}. In addition, we will need several domination-related lemmas to determine survival rules for component fixed points in the cyclic union, which we present next.
	
	\subsubsection{New domination-based results}\label{sec:domination-results} The proof of Theorem \ref{thm:generalized-cyclic-union} will proceed by induction on the number of nodes $n$, for a fixed number of components $N$.  The lemmas in this section are required to establish the base case where $n=N$, and to handle the inductive step. For the base case $G$ is an $n$-cycle and, as we show next, the only fixed point is the one with full support, $\FP(G, \bm{\varepsilon}, \bm{\delta}) = \{[n]\}$. The proof of this fact requires several small steps (lemmas), of interest on their own.
	
	First, we need to understand what happens in an ``open'' cyclic union: a composite linear chain.
	\begin{lemma}\label{lemma:reduced-linear-chain}
		Let $G$ be a \emph{composite linear chain}, that is, a graph with a partition $\tau_1,\dots\tau_N$ of its $n$ nodes, such that there is edges from every node in $\tau_i$ to every node in $\tau_{i+1}$. Then $\FP(G, \bm{\varepsilon}, \bm{\delta}) = \FP(G|_{\tau_N}, \bm{\varepsilon}_{\tau_N}, \bm{\delta}_{\tau_N})$.
	\end{lemma}
	\begin{proof}
		By the structure of the graph, all nodes in $\tau_1$ are dominated by some node in $\tau_2$, and thus are removable. Continuing this process of removals results in $\FP(G, \bm{\varepsilon}, \bm{\delta}) = \FP(G|_{\tau_N}, \bm{\varepsilon}_{\tau_N}, \bm{\delta}_{\tau_N})$.
	\end{proof}

	This lemma shows that the fixed points of a composite linear chain reduce to the last component. This can be seen in the example of Figure \ref{fig:composite_linear_chain}, where the activity is initialized in the first component, and then it trickles down through the chain to eventually stabilize in the last component, which happens to be a clique union of a single node and a 3-cycle. Thus, the attractor corresponding to $\{7,8,9,10\} \in \FP(G|_{\tau_4}, \bm{\varepsilon}_{\tau_4}, \bm{\delta}_{\tau_4})$ it's a fusion of the 7-8-9 repeating sequence, with the more stable activity of node 10.
	\begin{figure}[!h]
		\begin{center}
			\includegraphics[width=\textwidth]{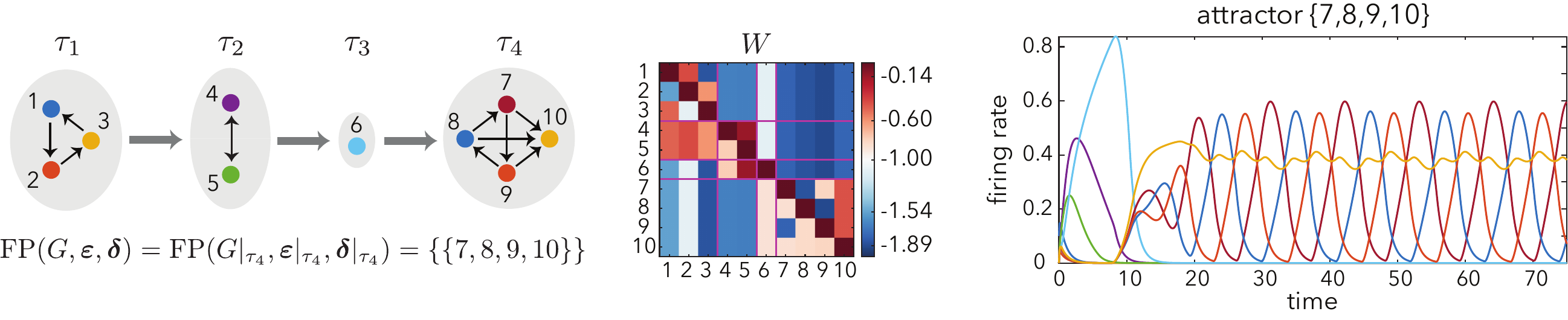}
		\end{center}
		\caption{\textbf{Composite linear chain.} A graph that is a composite linear chain and the connectivity matrix defined by it as a gCTLN. The fixed point supports reduce to those of the last component $\tau_4$. When the network is initialized in the first component, the activity trickles down the chain to finally settle in the last component.
			Adapted from \cite{parmelee2022a}.
		}
		\label{fig:composite_linear_chain}
	\end{figure}
	
	Next, we prove that the only singleton fixed point supports in any gCTLN correspond to sinks in the graph.
	
	\begin{lemma}\label{lemma:sinks}
		In a gCTLN, $\{k\} \in \FP(G, \bm{\varepsilon}, \bm{\delta})$ if and only if $k$ is a sink in $G$.
	\end{lemma}
	\begin{proof}
		By definition, $\{k\} \in \FP(G, \bm{\varepsilon}, \bm{\delta})$ if and only if there exists a fixed point $x^*$ with $x^*_k = \theta$ and $x^*_i = 0$ for all $i \neq k$. At a fixed point, $\left.\frac{dx}{dt}\right|_{x^*} = 0$. Putting those together we get:
		\begin{align*}
			0 &= \frac{dx_k}{dt} = -x^*_k + \left[\sum_{j \in [n]} W_{kj} x^*_j + \theta \right]_+ 
			= -\theta + [\theta]_+ = 0, \\
			0 &= \frac{dx_i}{dt} = -x^*_i + \left[\sum_{j \in [n]} W_{ij} x^*_j + \theta \right]_+ 
			= \left[W_{ik} \theta + \theta \right]_+, \quad \text{for } i \neq k.
		\end{align*}
		
		While the first equation always holds, the second gives the constraint:
		\[
		\left[W_{ik} \theta + \theta \right]_+ = 0 
		\quad \Leftrightarrow \quad W_{ik} \theta + \theta \leq 0 
		\quad \Leftrightarrow \quad W_{ik} \leq -1.
		\]
		
		Therefore, $\{k\} \in \FP(G, \bm{\varepsilon}, \bm{\delta})$ if and only if $W_{ik} \leq -1$ for all $i \neq k$, which means that $k$ is a sink.
	\end{proof}
	
	With these two results, we can easily prove the base case for our induction: the only fixed point of a cycle is the one with full support.
	
	\begin{lemma}\label{lemma:fp-generalized-cycle}
		Let $G$ be an $n$-cycle. Then $\FP(G, \bm{\varepsilon}, \bm{\delta}) = \{[n]\}$.
	\end{lemma}
	\begin{proof}
		Let $\sigma \in \FP(G, \bm{\varepsilon}, \bm{\delta}) $ and suppose for a contradiction that $\sigma \neq [n]$. This means that there exists some $k \in [n] \setminus \sigma$. Thus by Corollary \ref{cor:inheritance} we get that $\sigma \in \FP(G|_{[n]\setminus \{k\}},\bm{\varepsilon}_{[n]\setminus \{k\}}, \bm{\delta}_{[n]\setminus \{k\}})$. But $G|_{[n]\setminus \{k\}}$ is a linear chain (composite linear chain with components of size one) and thus by Lemma \ref{lemma:reduced-linear-chain} its set of fixed points reduces to that of the last component before $k$, i.e., $\sigma \in \FP(G|_{[n]\setminus \{k\}},\bm{\varepsilon}_{[n]\setminus \{k\}}, \bm{\delta}_{[n]\setminus \{k\}}) = \FP(G|_{\{k-1\}},\bm{\varepsilon}_{\{k-1\}}, \bm{\delta}_{\{k-1\}})$. This means that $\sigma = \{k-1\}$, for $\sigma \in \FP(G, \bm{\varepsilon}, \bm{\delta})$. But by Lemma \ref{lemma:sinks}, a singleton can only be a fixed point of the corresponding gCTLN if it's a sink, which $\{k-1\}$ is not in the larger graph (the $n$-cycle), reaching the desired contradiction. 
	\end{proof}
	
	We only need one more domination-related lemma for the main proof of Theorem \ref{thm:generalized-cyclic-union}. This lemma shows that any fixed point of a gCTLN defined as a cyclic union must have nonempty intersection with every component.
	
	\begin{lemma}\label{lemma:non-empty-intersections}
		Let $G$ be a cyclic union of components $\tau_1, \dots, \tau_N$, and let $(G, \bm{\varepsilon}, \bm{\delta})$ be a gCTLN. Then
		$$
		\sigma \in \FP(G, \bm{\varepsilon}, \bm{\delta}) \quad \Rightarrow \quad \sigma \cap {\tau_i} \neq \emptyset \quad \forall i \in [N].
		$$
	\end{lemma}
	\begin{proof} 
		Suppose that $ \sigma \cap {\tau_i}  = \emptyset$ for some $i$. This implies that $\sigma \subseteq [n] \setminus \tau_i$ and thus by Corollary \ref{cor:inheritance}, $\sigma \in \FP(G|_{[n] \setminus \tau_i},\bm{\varepsilon}_{[n] \setminus \tau_i},\bm{\delta}_{[n] \setminus \tau_i})$. But $G|_{[n] \setminus \tau_i}$ is a linear chain and thus by Lemma \ref{lemma:reduced-linear-chain}, it follows that $\FP(G|_{[n] \setminus \tau_i},\bm{\varepsilon}_{[n] \setminus \tau_i},\bm{\delta}_{[n] \setminus \tau_i}) =  \FP(G|_{\tau_{i-1}},\bm{\varepsilon}_{\tau_{i-1}}, \bm{\delta}_{\tau_{i-1}})$. This implies $\sigma \subseteq \tau_{i-1}$. However, $$\sigma \notin \FP(G|_{\tau_{i-1}\cup\tau_{i}},\bm{\varepsilon}_{\tau_{i-1}\cup\tau_{i}},\bm{\delta}_{\tau_{i-1}\cup\tau_{i}}),$$ because nodes in $\tau_{i}$ dominate the nodes of $\sigma$. But then Corollary \ref{cor:inheritance} implies that $\sigma \notin \FP(G, \bm{\varepsilon}, \bm{\delta})$, reaching a contradiction. 
	\end{proof}
	
	\subsubsection{Cyclic union proof}\label{sec:cyclic-union-proof}
	
	We are now in position to prove Theorem \ref{thm:generalized-cyclic-union}, which we restate below for clarity:
	\begingroup
	\def\thetheorem{\ref{thm:generalized-cyclic-union}} 
	\begin{theorem}[generalized cyclic union]
		Let $G$ be a cyclic union of components $\tau_1, \dots, \tau_N$, and let $(G, \bm{\varepsilon}, \bm{\delta})$ be a gCTLN. Then
		$$ \sigma \in \FP(G, \bm{\varepsilon}, \bm{\delta}) \quad \Leftrightarrow \quad \sigma = \bigcup_{\ell=1}^{N} \sigma_\ell, \text{ with } \sigma_\ell \in \FP(G|_{\tau_\ell}, \bm{\varepsilon}_{\tau_\ell}, \bm{\delta}_{\tau_\ell}).$$
		Moreover, if $\sigma \in \FP(G, \bm{\varepsilon}, \bm{\delta})$, then
		$$ \idx(\sigma) = \prod_{i=1}^N \idx(\sigma \cap \tau_i). $$
	\end{theorem}

	\addtocounter{theorem}{-1}
	\endgroup
	\begin{proof}
		We proceed by complete induction on the number of nodes $n$ in the network, for a fixed number of components $N$. For the base case where $n = N$, $G$ is an $n$-cycle. By Lemma \ref{lemma:fp-generalized-cycle}, $\FP(G, \bm{\varepsilon}, \bm{\delta}) = \{[n]\}$. Because every component is a single vertex, the result trivially holds. The index formula also follows trivially from this observation.
		
		Now assume the result holds for all cyclic unions with $N$ components and $m$ vertices, where $N \leq m < n$. Because all components of a generalized cyclic union are simply-embedded in the larger network, by Lemmas \ref{lemma:weak-fixed-point-partition} and \ref{lemma:non-empty-intersections}, we get:
		\begin{align*}
			\sigma \in \FP(G, \bm{\varepsilon}, \bm{\delta}) 
			&\Leftrightarrow \sigma \in \FP(G|_\sigma, \bm{\varepsilon}_\sigma, \bm{\delta}_\sigma) \\
			&\quad \text{and } \sigma \cap \tau_i \in \FP(G|_{\tau_i}, \bm{\varepsilon}_{\tau_i}, \bm{\delta}_{\tau_i}) 
			\quad \text{for all } i \in [N].
		\end{align*}
		
		Thus, it suffices to show that if we have $\sigma \cap \tau_i \in \FP(G|_{\tau_i}, \bm{\varepsilon}_{\tau_i}, \bm{\delta}_{\tau_i})$ for all $i \in [N]$, then we also have $\sigma \in \FP(G|_\sigma, \bm{\varepsilon}_\sigma, \bm{\delta}_\sigma)$.
		
		\underline{Case I}: If $|\sigma| < n$: suppose $\sigma \cap \tau_i \in \FP(G|_{\tau_i}, \bm{\varepsilon}_{\tau_i}, \bm{\delta}_{\tau_i})$ for all $i \in [N]$. By Corollary \ref{cor:inheritance}, we obtain $\sigma \cap \tau_i \in \FP(G|_{\sigma \cap \tau_i}, \bm{\varepsilon}_{\sigma \cap \tau_i}, \bm{\delta}_{\sigma \cap \tau_i})$. Applying the inductive hypothesis to $G|_{\sigma}$ we get that $ \sigma \in \FP(G|_\sigma, \bm{\varepsilon}_\sigma, \bm{\delta}_\sigma)$, as desired. The index formula holds because $\idx(\sigma)$ is the same as what it was in the smaller graph $G|_{\sigma}$ (the index only depends on $G|_{\sigma}$), and thus it is given by the inductive hypothesis.
		
		\underline{Case II}: If $|\sigma| = n$. Then $G|_\sigma = G$ and $\sigma \cap \tau_i = \tau_i$ for all $i \in [N]$. By Lemma \ref{lemma:weak-fixed-point-partition}, $\sigma \in \FP(G, \bm{\varepsilon}, \bm{\delta}) \Rightarrow \sigma \cap \tau_i \in \FP(G|_{\tau_i}, \bm{\varepsilon}_{\tau_i}, \bm{\delta}_{\tau_i})$ for all $i \in [N]$. To show the converse direction, suppose $\tau_i \in \FP(G|_{\tau_i}, \bm{\varepsilon}_{\tau_i}, \bm{\delta}_{\tau_i})$ for all $i \in [N]$.	If $\sigma = [n] \notin \FP(G, \bm{\varepsilon}, \bm{\delta})$, then $|\FP(G, \bm{\varepsilon}, \bm{\delta})| = \prod_{i \in [N]} |\FP(G|_{\tau_i}, \bm{\varepsilon}_{\tau_i}, \bm{\delta}_{\tau_i})| - 1$,	since all the smaller elements of $\FP(G, \bm{\varepsilon}, \bm{\delta})$ are indeed given  by picking a fixed point support from each component graph $G|_{\tau_i}$.  By parity, each $|\FP(G|_{\tau_i}, \bm{\varepsilon}_{\tau_i}, \bm{\delta}_{\tau_i})|$ is odd, thus the product is odd. This implies that $|\FP(G, \bm{\varepsilon}, \bm{\delta})|$ even, which contradicts parity for $G$. Thus $\sigma \in \FP(G, \bm{\varepsilon}, \bm{\delta})$.
		
		To show that the index formula holds for $\sigma = [n]$, using Theorem \ref{thm:parity} for $\FP(G, \bm{\varepsilon}, \bm{\delta})$, we compute:
		\begin{align*}
			1 &= \sum_{\sigma \in \FP(G, \bm{\varepsilon}, \bm{\delta})} \idx(\sigma) = \sum_{\sigma \in \FP(G, \bm{\varepsilon}, \bm{\delta}) \setminus \{[n]\}} \idx(\sigma) + \idx([n]) \\ 
			&= \sum_{\sigma \in \FP(G, \bm{\varepsilon}, \bm{\delta}) \setminus \{[n]\}} \left(\prod_{i \in [N]} \idx(\sigma \cap \tau_i)\right) + \idx([n]) \\
			&= \prod_{i \in [N]} \left(\sum_{\omega \in \FP(G|_{\tau_i}, \bm{\varepsilon}_{\tau_i}, \bm{\delta}_{\tau_i})} \idx(\omega)\right) - \prod_{i \in [N]} \idx(\tau_i) + \idx([n]) \\
			&= 1 - \prod_{i \in [N]} \idx(\tau_i) + \idx([n]).
		\end{align*}
		Thus,
		$$
		\idx([n]) = \prod_{i \in [N]} \idx(\tau_i),
		$$
		as desired.
	\end{proof}
	
	Note that Theorem \ref{thm:low-rank-gluing-menu} was used implicitly in this proof through the application of Lemma \ref{lemma:weak-fixed-point-partition}.

	\subsection{Proof of Theorem \ref{thm:generalized-disjoint-clique-union}: clique and disjoint unions}\label{sec:clique-disjoint-union}
	
	Recall the clique and disjoint union constructions from Fig \ref{fig:CTLN_results_summary}C-D. Formally, as defined in \cite{curto2019a}:
	\begin{definition}
		Let $\{\tau_1, \ldots, \tau_N\}$ be a partition of the index set of a directed graph $G$. We say that $G$ is a:
		\begin{enumerate}
			\item \textit{clique union} of the induced subgraphs $G|_{\tau_1}, \ldots, G|_{\tau_N}$ if it contains all edges between vertices in distinct components $\tau_i, \tau_j$ for all $i, j \in [N]$.
			\item \textit{disjoint union} of the induced subgraphs $G|_{\tau_1}, \ldots, G|_{\tau_N}$ if there are no edges between distinct components $\tau_i, \tau_j$.
		\end{enumerate}
	\end{definition}
	
	A gCTLN defined by these graphs is a rank-1 gluing. A simple application of Theorem~\ref{thm:rank-1-gluings}, along with one elementary survival lemma for fixed points in gCTLNs will yield Theorem \ref{thm:generalized-disjoint-clique-union}. This survival lemma will allow us to determine which case of Theorem \ref{thm:rank-1-gluings}, (a) survivors or (b) dyers, applies to disjoint and clique unions. This lemma does not require domination, though it uses similar proof arguments. It relies directly on the gCTLN weight prescription to check the survival/death of a fixed point support when a new node $k$ is added. 
	
	\begin{lemma}\label{lemma:single-node-survival}
		Let $(G, \bm{\varepsilon}, \bm{\delta})$ be a gCTLN, and let $\sigma \in \FP(G|_\sigma, \bm{\varepsilon}_\sigma, \bm{\delta}_\sigma)$ with $\sigma \neq \emptyset$. Let $k \notin \sigma$. Then 
		\begin{enumerate}
			\item if, for all $i \in \sigma$, we have $i \rightarrow k$, then $\sigma \notin \FP(G|_{\sigma \cup \{k\}}, \bm{\varepsilon}_{\sigma \cup \{k\}}, \bm{\delta}_{\sigma \cup \{k\}})$ ($\sigma$ dies).
			\item if, for all $i \in \sigma$, we have $i \not\rightarrow k$, then $\sigma \in \FP(G|_{\sigma \cup \{k\}}, \bm{\varepsilon}_{\sigma \cup \{k\}}, \bm{\delta}_{\sigma \cup \{k\}})$ ($\sigma$ survives).
		\end{enumerate}
	\end{lemma}
	
	Anytime we have $i \rightarrow k$ for all $i \in \sigma$, we say that $k$ is a \textit{target} of $\sigma$. The essence of the proof is to compare the input received by the new node $k \notin \sigma$ with the input received by a node $j \in \sigma$. In case (i), $k$ is strongly inhibited by the active neurons in $\sigma$, ensuring it remains off. In case (ii), $k$ receives weak inhibition from $\sigma$ (at least as much input as any $j \in \sigma$). This guarantees $k$ will turn on, then ``killing'' the fixed point $\sigma$.
	
	\begin{proof}
		Define
		\[
		y_i(x) = \sum_{j=1}^n W_{ij} x_j + \theta.
		\]
		Note that $y_i$ depends on the (sub)network under consideration, as determined by the range of indices in the sum (full network in the definition above). With this notation, the equations for the dynamics of the network become
		\[
		\frac{dx_i}{dt} = -x_i + [y_i(x)]_+ \quad \text{for } i \in [n].
		\]
		Thus, at a fixed point $x^*$ we have $x_i^* = [y_i^*]_+$, where $y_i^* = y_i(x^*)$.
		
		To prove the lemma, it suffices to verify that for (i), $y_k^* < 0$ (so $k$ remains off) and for (ii) $y_k^* > 0$ (so $k$ turns on), at the fixed point $x^*$ supported on $\sigma \in \FP(G|_\sigma, \bm{\varepsilon}_\sigma, \bm{\delta}_\sigma)$. Let $j \in \sigma$ (which must be nonempty, since the empty set cannot support a fixed point).
		
		For (a),  by assumption we have that for all $i \in \sigma$, 
		\[
		W_{ki} = -1 + \varepsilon_i \geq W_{ji} = 
		\begin{cases}
			-1 + \varepsilon_i & \text{if } i \rightarrow j \text{ in } G|_\sigma, \\
			-1 - \delta_i & \text{if } i \not\rightarrow j \text{ in } G|_\sigma.
		\end{cases}
		\]
		Then
		\begin{align*}
			y_k^* &= \sum_{i \in \sigma} W_{ki} x_i^* + \theta = \sum_{i \in \sigma \setminus \{j\}} W_{ki} x_i^* + W_{kj} x_j^* + \theta  \geq \sum_{i \in \sigma \setminus \{j\}} W_{ji} x_i^* + W_{kj} x_j^* + \theta \\
			&= \left( \sum_{i \in \sigma \setminus \{j\}} W_{ji} x_i^* + \theta \right) + W_{kj} x_j^*  = y_j^* + W_{kj} x_j^*.
		\end{align*}
		Since $j \in \sigma$ and $\sigma \in \FP(G|_\sigma, \bm{\varepsilon}_\sigma, \bm{\delta}_\sigma)$, we have $y_j^* = x_j^*$. Therefore,
		\[
		y_k^* \geq (1 + W_{kj}) x_j^* = \varepsilon_j x_j^* > 0.
		\]
		
		For (b),  the assumption gives instead that for all $i \in \sigma$,
		\[
		W_{ki} = -1 - \delta_i \leq W_{ji} = 
		\begin{cases}
			-1 + \varepsilon_i & \text{if } i \rightarrow j \text{ in } G|_\sigma, \\
			-1 - \delta_i & \text{if } i \not\rightarrow j \text{ in } G|_\sigma.
		\end{cases}
		\]
		Then
		\begin{align*}
			y_k^* &= \sum_{i \in \sigma} W_{ki} x_i^* + \theta  = \sum_{i \in \sigma \setminus \{j\}} W_{ki} x_i^* + W_{kj} x_j^* + \theta  \leq \sum_{i \in \sigma \setminus \{j\}} W_{ji} x_i^* + W_{kj} x_j^* + \theta \\
			&= \left( \sum_{i \in \sigma \setminus \{j\}} W_{ji} x_i^* + \theta \right) + W_{kj} x_j^*  = y_j^* + W_{kj} x_j^*.
		\end{align*}
		Again, since $y_j^* = x_j^*$, we have
		\[
		y_k^* \leq (1 + W_{kj}) x_j^* = -\delta_j x_j^* < 0.
		\]
	\end{proof}
	
	We can now prove Theorem \ref{thm:generalized-disjoint-clique-union}. The proof is a straightforward application of Theorem \ref{thm:rank-1-gluings} and Lemma \ref{lemma:single-node-survival}. To apply Theorem \ref{thm:rank-1-gluings} it suffices to show that for a disjoint union all component fixed points survive, and for a clique union all component fixed points die. 
	
	Indeed, for clique unions, any node $k$ in any other component $\tau_j$ ($j \neq i$) is a target of $\sigma_i$ (by the definition of a clique union). Lemma \ref{lemma:single-node-survival}(i) guarantees that every component fixed point support $\sigma_i \in \FP(G|_{\tau_i}, \bm{\varepsilon}_{\tau_i}, \bm{\delta}_{\tau_i})$ dies in the full network. Thus, all component fixed point supports die ($\sigma_i \in D_{\tau_i}$) and Theorem \ref{thm:rank-1-gluings}(a) applies, giving the result.
	
	Similarly, for disjoint unions, take any component fixed point $\sigma_i \in \FP(G|_{\tau_i}, \bm{\varepsilon}_{\tau_i}, \bm{\delta}_{\tau_i})$ and consider any node $k \notin \tau_i$. By the definition of a disjoint union, there are no edges from $\sigma_i$ to $k$ and thus Lemma \ref{lemma:single-node-survival}(ii) applies. This holds for all $k \notin \tau_i$, so $\sigma_i$ survives the addition of all nodes in all other components. Thus, every component fixed point support is a survivor ($\sigma_i \in S_{\tau_i}$) and Theorem \ref{thm:rank-1-gluings}(b) applies, giving the result.
	
	This proves the following:
	\begingroup
	\def\thetheorem{\ref{thm:generalized-disjoint-clique-union}}
	\begin{theorem}[generalized clique and disjoint union]
		Let $G$ be a graph with a partition of its vertices $\{\tau_1 \mid \cdots \mid \tau_N\}$, and let $(G, \bm{\varepsilon}, \bm{\delta})$ define a gCTLN.
		\begin{enumerate}
			\item If $G$ is a clique union of $G|_{\tau_1}, \ldots, G|_{\tau_N}$, then
			$$ \sigma \in \FP(G, \bm{\varepsilon}, \bm{\delta}) \quad \Leftrightarrow \quad \sigma = \bigcup_{\ell=1}^{N} \sigma_\ell, \text{ with } \sigma_\ell \in \FP(G|_{\tau_\ell}, \bm{\varepsilon}_{\tau_\ell}, \bm{\delta}_{\tau_\ell}).$$
			\item If $G$ is a disjoint union of $G|_{\tau_1}, \ldots, G|_{\tau_N}$, then 
			$$\sigma \in \FP(G, \bm{\varepsilon}, \bm{\delta}) \quad \Leftrightarrow \quad \sigma = \bigcup_{\ell=1}^{N} \sigma_\ell, \text{ with }  \sigma_\ell \in \FP(G|_{\tau_\ell}, \bm{\varepsilon}_{\tau_\ell}, \bm{\delta}_{\tau_\ell}) \cup \{\emptyset\}$$
		\end{enumerate}
	\end{theorem}
	\addtocounter{theorem}{-1}
	\endgroup
	
	Note that the requirement in Theorem \ref{thm:generalized-disjoint-clique-union}(a) for a \textit{full} clique union is strict, as this avoids domination relationships to exist in the graph. To illustrate this, note that in both networks of Figure \ref{fig:not_clique_not_disjoint_union}, each component fixed point dies in the full graph, for both A and B, by Lemma \ref{lemma:single-node-survival}. In panel A, it's because $4$ is a target of $\{1,2,3\}$ and $2$ is a target of $\{4,5\}$. In panel B it's because, as a clique union, all nodes have targets in the other component.
	\begin{figure}[!h]
		\centering
		\includegraphics[width=0.75\textwidth]{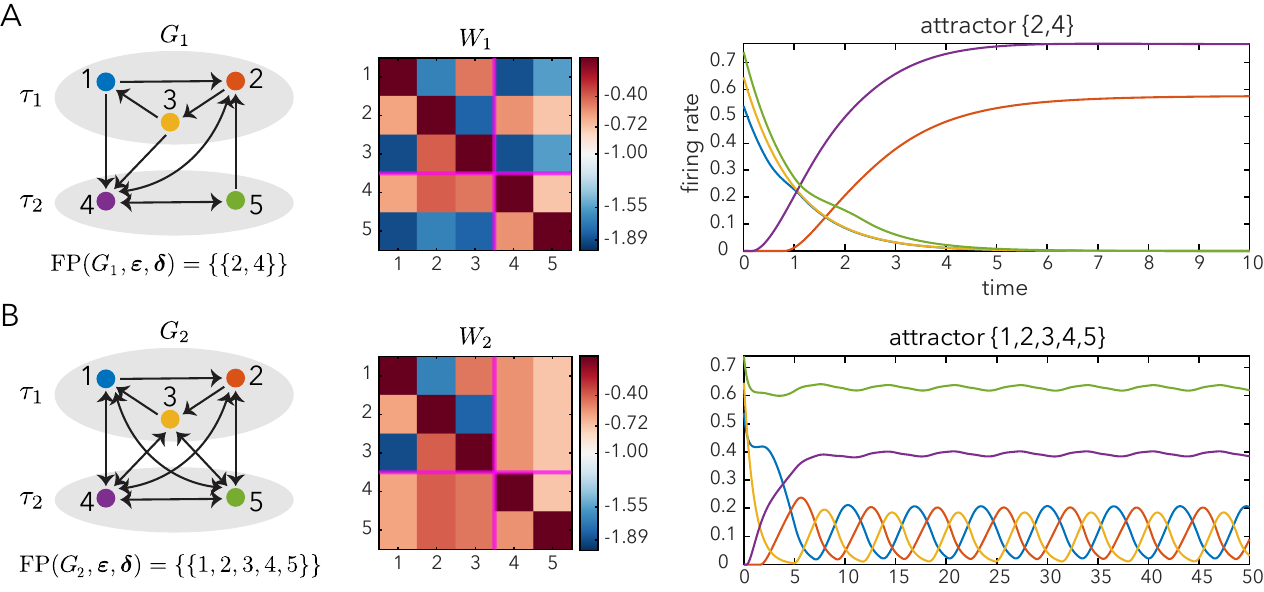}
		\caption{\textbf{Comparison of a graph that has only some targets in the other components and a clique union, defined on the same component subgraphs.} 
			(A) A graph $G_1$ where $4$ is a target of  $\{1,2,3\}$ and $2$ is a target of $\{4,5\}$, along with a gCTLN connectivity matrix and the attractor corresponding to the support $\{2,4\}$.
			(B) A graph $G_2$ that is a clique union, where all nodes have targets in the other component, along with a gCTLN connectivity matrix and the attractor corresponding to the support $\{\{1,2,3,4,5\}\}$.
		}
		\label{fig:not_clique_not_disjoint_union}
	\end{figure}
	
	However, the set of fixed point supports of the two networks are different. Because the graph of Figure \ref{fig:not_clique_not_disjoint_union}A is not a full clique union, it has domination in it (2 dominates 5, 4 dominates 1, and 4 dominates 3). Therefore, Theorem \ref{thm:removal-of-dominated-nodes} applies, and the fixed points reduce to $\FP(G_1, \bm{\varepsilon}, \bm{\delta}) = \{\{2,4\}\}$. Because $G|_{\{2,4\}}$ is a clique, we expect this support to correspond to a stable fixed point, with 2 and 4 high firing, as the rate curves of Figure \ref{fig:not_clique_not_disjoint_union}A confirm.
	
	By contrast, in the clique union (Fig. \ref{fig:not_clique_not_disjoint_union}B), none of the targets dominate any nodes because all edges between $\tau_1$ and $\tau_2$ are bidirectional. Additionally, by Theorem \ref{thm:generalized-disjoint-clique-union} we have $\FP(G_2, \bm{\varepsilon}, \bm{\delta}) = \{\{1,2,3,4,5\}\}$, and the corresponding attractor is a fusion of a sequential attractor (supported on $\{1,2,3\}$) and a stable fixed point (supported on $\{4,5\}$).

\section{Discussion}\label{sec:discussion}

We have presented a formal mathematical link between modular network structure and compositionality, proving that the global fixed points of low-rank gluings are constrained to be compositions of local fixed points. For rank-1 gluings, our theory fully characterizes which combinations of local fixed points yield global ones. Biologically, this translates to the fact that when individual neurons in a source region project uniformly to all neurons in a target area (though weights can vary across different target regions), the functional repertoire of a multi-area network is not arbitrary but constrained to combinations of primitives encoded within each component module.

Although our formal results address compositionality at the level of fixed points, applying these rules to gCTLNs allows us to empirically transfer them to the network's attractors. For cyclic, clique, and disjoint unions, explicit survival rules for local fixed points enable a clear decomposition of global fixed points, and consequently of attractors, offering a reliable recipe to produce families of networks supporting sequential, fusion, and component-wise winner-take-all attractors. This in turn showed how the conservation of a primitive's firing pattern depends on both the directionality of inter-regional couplings and context (which other primitives are concurrently engaged). Specifically, cyclic inter-area connectivity will turn primitives into oscillations during compositional tasks (Figure \ref{fig:attractor_correspondance_cyclic_union}), whereas predominantly bidirectional connectivity tends to preserve the firing patterns of primitives more (Figure \ref{fig:fixed_point_compositionality_and_attractors}). In both cases, the precise characteristics of each neuron's firing within a compositional task will depend on the set of primitives concurrently engaged.
	
Moreover, the extension of the gluing rules from CTLNs to gCTLNs demonstrates that functional compositionality is more robust than originally established in \cite{curto2019a,parmelee2022a}, suggesting that the core computational function is primarily encoded in the graph structure rather than in the precise values of the synapses. If we interpret the binary synapses of CTLNs as an idealized structural blueprint, gCTLNs would represent a noisier, more physiologically realistic instantiation. In this view, the quantitative variations we observe in individual firing rates might mirror the animal-to-animal circuit variability widely observed in nature \cite{long2008, tang2010}.

	\paragraph{On the low-rank assumption.}
	
	To mathematically formalize how a modular network supports compositionality, we have imposed the assumption that intercomponent connectivity is low-rank. Low-rank networks have been widely studied in computational neuroscience \cite{mastrogiuseppe2018a, schuessler2020, beiran2021a,clark2025}, often connected to the idea that the brain uses low-dimensional representations of tasks within high-dimensional state spaces. However, while our framework supposes low-rank intercomponent connectivity, the global connectivity matrix is not low-rank itself, as each diagonal block might be full rank. Yet, it naturally gives rise to a variety of low-dimensional attractor dynamics. In our case, this single, presumably full-rank matrix yields a combinatorial explosion of attractors that coexist for fixed values of the parameters, all accessible via different initial conditions. Such coexistence allows the system to easily transition between attractors via transient external pulses. This aligns with the contemporary view that the brain operates as a high-dimensional nonlinear dynamical system, where experimentally observed low-dimensional manifolds represent task-specific subspaces within a much larger state space \cite{amematsro2025,gallego2017,sussillo2014,duncker2021,perich2025,vyas2020,schmutz2026}. 

	While general low-rank networks abound in the literature, the study of modular networks with low-rank intercomponent connectivity has only recently garnered interest in neuroscience as a model for low-dimensional communication subspaces between brain regions \cite{clark2025}. Though our structural constraints are similar to those proposed by Clark and Beiran in \cite{clark2025}, there are key differences. Structurally, their components are combinations of random and low-rank matrices, whereas our components are internally unconstrained. Additionally, our intercomponent connectivity enforces a stricter low-rank condition where one of the outer-product vectors is fixed to be an all-ones vector. Functionally, while their work concerns the flow of information between regions, the focus of our work is on the compositionality of the dynamics. Exploring how these two frameworks might interact, either by analyzing information flow within our structures or by studying fixed point compositionality within theirs, remains a direction for future research.
		
	\paragraph{Mathematical biology parallels.}
	
	The central role of the underlying graph structure was particularly evident in our Clione example, where similar adjacency matrices yielded qualitatively similar sequential dynamics, even with different defining systems ODEs. A similar parallel appears in related work by the same group of authors, Rabinovich et al., \cite{rabinovich2010,rabinovich2018}, who proposed generalized Lotka-Volterra models of feature binding via heteroclinic channels to produce sequential trajectories. Again, while similar to our cyclic union construction, the specific mechanisms for obtaining sequential trajectories representing distinct modalities differ: the work of Rabinovich et al. relies on sequences of unstable saddle fixed points (heteroclinic orbits), whereas our constructions yield limit cycle trajectories. Additionally, because our unconstrained components support diverse internal motifs rather than being restricted to single-neuron firing sequences, our model offers an expressive architecture for complex sequential patterns.
	
	Remarkably, this tight link between structural modularity and functional compositionality also cuts across from continuous to discrete dynamical systems. Our results parallel those of Kadelka et al. \cite{kadelka2023}, who formally explored this link in gene regulatory networks modeled as Boolean systems. In those discrete models, the structural decomposition of a network into maximal strongly connected components induces a decomposition of its dynamic structure, endowing the system with robustness and a rich dynamic repertoire. Their results, together with our work, suggest that compositionality might be general property of certain modular systems and might point to a broader category theoretic phenomenon \cite{carranza2025,curto2023}.
	
	Finally, this structural centrality is also illustrated by the fact that for every inhibition-dominated gCTLN, an equivalent excitatory-inhibitory (E-I) network can be defined on the same graph such that there is a one-to-one mapping between their fixed points. Moreover, computational results show that the dynamics of the gCTLN and its equivalent E-I network also match \cite{curto2025a}, provided inhibition is fast enough. CTLN graph-to-dynamics principles have also been shown to hold for clustered inhibition-stabilized Hawkes networks with fast inhibition \cite{lienkaemper2025a}.
	
	Taken together, these parallels show how network structure plays a fundamental role in governing network dynamics, even across frameworks as dissimilar as discrete and continuous time, and in phenomena as complex as fixed point compositionality.
		
	\paragraph{Open questions.} 
	
	Several open questions remain. First, while our theorems completely specify global fixed points, a comprehensive analysis of their stability is still pending, and formally proving that the trajectories shown here are attractors remains a non-trivial challenge \cite{bel2021}. As such, our observations regarding limit cycle compositionality remain empirical.
		
	Second, the exact boundaries of our structural constraints remain to be explored in two main directions. First, by applying perturbations to the graph $G$ (adding or deleting edges) to evaluate the minimal number of directed edges required to support the types of compositional attractors exemplified here. While this issue has been studied for CTLNs \cite{parmelee2022a}, how those results extend to gCTLNs and general TLNs is unknown. A second direction should test the sensitivity of the theoretical and computational results to synaptic noise by perturbing the requirement that a given node must project to all neurons within a target module with identical weight.
	
	Third, to fully exploit the potential of Theorem \ref{thm:rank-1-gluings}, a more extensive theory of ``survival rules'' for gCTLNs and TLNs in general is needed. Indeed, Theorem \ref{thm:rank-1-gluings} became significantly clearer when applied to gCTLNs precisely because many of their local survival rules were already known. While some of our examples demonstrate that general TLN survival rules can behave remarkably similar to those of CTLNs, we have shown, for instance, that Theorem \ref{thm:uniform-in-deg} does not hold for gCTLNs. This limitation is also illustrated in Figure \ref{fig:fixed_point_compositionality_and_attractors}B, where a CTLN defined on the same graph would instead yield $\FP(G_2) = \{ \{8\}, \{5, 6, 7\}, \{5, 6, 7, 8\} \}$.
	
	Nevertheless, our framework offers a step forward in the study of nonlinear network dynamics. In particular, even though threshold-linear networks are highly idealized as brain models, they remain widely used across computational neuroscience, and our results are a testament to their expressiveness. This is the case even under the additional simplifications imposed here, such as assuming uniform inputs, uniform timescales, and an inhibition-dominated regime. These simplifications serve to isolate the role of connectivity in enabling compositionality, providing a bare-bones scaffolding or structural blueprint on which the richer complexities of biological brains might rest.
		
	\paragraph{Acknowledgments.} The author thanks Professors Carina Curto, Katie Morrison and Chris Rose for guidance and helpful comments on the manuscript. Part of this work was supported by the National Science Foundation under Grant No. DMS-1929284 while the author was in residence at the Institute for Computational and Experimental Research in Mathematics (ICERM) in Providence, RI, during the Math + Neuroscience program. The findings and conclusions in this article do not necessarily reflect the view of the funding agencies. Gemini's Writing Editor was used to revise the text for spelling, grammar, and clarity. The author has reviewed all AI-generated suggestions and assumes responsibility for all content.
	
	
	
	
	\paragraph{Data \& code availability.}	MATLAB code needed to reproduce all simulations in this manuscript is publicly available on GitHub at \url{https://github.com/juliana-londono/low-rank-gluings}.
	
	\bibliographystyle{unsrt}
	\bibliography{generalized_cyclic_unions.bib}
	\clearpage 
	
\end{document}